	\documentclass[11pt,a4paper]{article}
	\usepackage[english]{babel}
	\usepackage[utf8x]{inputenc}
	\usepackage[T1]{fontenc}
	\usepackage{xspace,authblk}
	\usepackage{footnote}
	\usepackage[a4paper,top=3cm,bottom=2cm,left=2cm,right=2cm,marginparwidth=1cm]{geometry}
\usepackage{hyperref}
\hypersetup{colorlinks, citecolor=red, filecolor=blue, linkcolor=blue, urlcolor=blue}

    \usepackage{amsmath,amssymb}
	\usepackage{amsthm,graphicx, color}
    \usepackage{enumerate}
	\usepackage{float}
    \usepackage{algorithm}
    \usepackage{algorithmic}
	\usepackage{multirow}
	\usepackage{color}
	\usepackage{comment}
	\usepackage{url}

	\numberwithin{equation}{section}
	\newtheorem{thm}{Theorem}
    \numberwithin{thm}{section} 
	\newtheorem{lemma}[thm]{Lemma}

	\newtheorem{assump}{Assumption}
	\newtheorem*{fact}{Fact}

    \allowdisplaybreaks
	
\newcommand{\eat}[1]{}

\newcommand{\sbk}{\mathsf{SBK}}

\newcommand{\optskp}{\mathsf{OPT}_{\mathsf{SKP}}(\mathcal{J})}

\newcommand{\optsbk}{\mathsf{OPT}_{\mathsf{SBK}}(\mathcal{J})}

\newcommand{\sk}{\mathsf{SKP}}
\newcommand{\probemax}{Probemax}
\newcommand{\pandora}{Pandora's Box}
\newcommand{\probek}{ProbeTop-$k$}

\newcommand{\skrp}{\ensuremath{\operatorname{SK-RP}}\xspace}

\def\etal{\emph{et~al.}\xspace}
\newcommand{\ie}{{\em i.e.,~\xspace}}
\newcommand{\eg}{{\em e.g.,~\xspace}}

\newcommand{\opt}{\mathsf{OPT}}

\newcommand{\nopt}{\widetilde{\mathsf{OPT}}}

\newcommand{\mopt}{\mathsf{MAX}}
\newcommand{\mmopt}{\max_{I\in\mathcal{V}}\mathrm{DP}_1(I,\cA)}

\newcommand{\bP}{\mathbb{P}}

\newcommand{\bT}{\mathbb{T}}
\newcommand{\bE}{\mathbb{E}}

\newcommand{\bC}{\mathbb{C}}
\newcommand{\sC}{\mathsf{C}}
\newcommand{\sCA}{\mathsf{CA}}

\newcommand{\cG}{\mathcal{G}}
\newcommand{\cV}{\mathcal{V}}
\newcommand{\cS}{\mathcal{S}}
\newcommand{\cJ}{\mathcal{J}}
\newcommand{\cA}{\mathcal{A}}
\newcommand{\cB}{\mathcal{B}}
\newcommand{\cT}{\mathcal{T}}
\newcommand{\cR}{\mathcal{R}}

\newcommand{\cD}{\mathcal{D}}
\newcommand{\cI}{\mathcal{I}}

\newcommand{\wt}[1]{\widetilde{#1}}

\newcommand{\ep}{\varepsilon}
\newcommand{\lam}{\lambda}
\newcommand{\lf}{\ensuremath{\operatorname{LF}}\xspace}

\newcommand{\seg}{\mathsf{seg}}
\newcommand{\DP}{\mathrm{DP}}
\newcommand{\dpp}{\mathcal{M}}
\newcommand{\sig}{\mathsf{Sg}}

\title{A PTAS for a Class of Stochastic Dynamic Programs
\footnote{
This work is published in the 45th International Colloquium on Automata, Languages, and Programming (ICALP2018) \cite{icalp2018}.
This research is supported in part by the National Basic Research Program of China Grant 2015CB358700, the National Natural Science Foundation of China Grant 61772297, 61632016, 61761146003, and a grant from Microsoft Research Asia.
}
}
\author[1]{Hao Fu \footnote{Email: fu-h13@mails.tsinghua.edu.cn}}
\author[1]{Jian Li \footnote{Email: lijian83@mail.tsinghua.edu.cn}}
\author[2]{Pan Xu \footnote{Email: panxu@cs.umd.edu}}
\affil[1]{
Institute for Interdisciplinary Information Sciences, 
Tsinghua University, Beijng, China.}
\affil[2]{
Department of Computer Science,  
University of Maryland, College Park, USA.}

\begin{document}
\maketitle


\begin{abstract}
We develop a framework for obtaining polynomial time approximation schemes (PTAS) for a class of stochastic dynamic programs. Using our framework, we obtain the first PTAS for 
the following stochastic combinatorial optimization problems:
\begin{enumerate}
\item {\em \probemax} \cite{Munagala2016}: We are given a set of $n$ items, each item $i\in [n]$ has a value $X_i$ which is an independent random variable with a known (discrete) distribution $\pi_i$. 
We can {\em probe} a subset $P\subseteq [n]$ of items sequentially.
Each time after  {probing} an item $i$, we observe its value realization, which follows the distribution $\pi_i$. 
We can {\em adaptively} probe at most $m$ items and each item can be probed at most once. The reward is the maximum among the $m$ realized values. Our goal is to design an adaptive probing policy such that the expected value of the reward is maximized.
To the best of our knowledge, the best known approximation ratio is $1-1/e$, due to Asadpour \etal~\cite{asadpour2015maximizing}.
We also obtain PTAS for some generalizations and variants of the problem.



\item {\em Committed \pandora} \cite{wei79,sahil2018}: We are given a set of $n$ boxes. For each box $i\in [n]$, the cost $c_i$ is deterministic and the value $X_i$ is an independent random variable with a known (discrete) distribution $\pi_i$.
Opening a box $i$ incurs a cost of $c_i$. 
We can adaptively choose to open the boxes (and observe their values) or stop.
We want to maximize the expectation of 
the realized value of the last opened box minus the total opening cost.

\item {\em Stochastic Target} \cite{ilhan2011adaptive}: 
Given a predetermined target $\bT$ and $n$ items, 
we can adaptively insert the items into a knapsack and insert at most $m$ items. 
Each item $i$ has a value $X_i$ which is an independent random variable with a known (discrete) distribution. 
Our goal is to design an adaptive policy such that the probability of the total values of all items inserted being larger than or equal to $\bT$ is maximized. We provide the first bi-criteria PTAS for the problem.

\item {\em Stochastic Blackjack Knapsack} \cite{levin2014adaptivity}:
We are given a knapsack of capacity $\bC$ and  
probability distributions of $n$ independent random variables $X_i$.
Each item $i\in[n]$ has a size $X_i$ and a profit $p_i$. We can adaptively insert the
items into a knapsack, as long as the capacity constraint is not violated.
We want to maximize the expected total profit of all inserted items.  
If the capacity constraint is violated, we lose all the profit.
We provide the first bi-criteria PTAS for the problem.

\end{enumerate}
\end{abstract}

\newpage

\section{Introduction} 

Consider an online stochastic optimization problem with a finite number of rounds. 
There are a set of tasks (or items, boxes, jobs or actions).
In each round, we can choose a task 
and each task can be chosen at most once. 
We have an initial ``state'' of the system
(called the value of the system). At each time period, we can select a task. Finishing the task generates some (possibly stochastic) feedback, including changing the value of the system
and providing some profit for the round. 
Our goal is to design a strategy to maximize our total (expected) profit.   

The above problem can be modeled as a class of stochastic dynamic programs which was introduced by Bellman \cite{bellman1957}. There are many problems in stochastic combinatorial optimization which fit in this model, \eg the stochastic knapsack problem \cite{dean2005adaptivity}, the \probemax~problem \cite{Munagala2016}. 
Formally, the problem is specified by
a $5$-tuple $(\cV,\cA,f,g,h,T)$. 
Here, $\cV$ is the set of all possible values of the system. 
$\cA$ is a finite set of items or tasks which can be selected and each item can be chosen at most once. This model proceeds for at most $T$ rounds.
At each round $t\in[T]$, 
we use $I_t\in\cV$ to denote the current value of the system
and $\cA_t \subseteq \cA$ the set of remaining available items.
If we select an item $a_t\in \cA_t$,
the value of the system changes to $f(I_t,a_t)$. Here $f$ may be stochastic and is assumed to be independent for each item $a_t\in \cA$.
Using the terminology from Markov decision processes,
the state at time $t$ is $s_t=(I_t,\cA_t)\in \cV\times2^{\cA}$.
\footnote{This is why we do not call $I_t$ the state of the system.}
Hence, if we select an item $a_t\in \cA_t$, the evolution of the state is determined by the state transition function $f$: 
\begin{equation}
s_{t+1}=(I_{t+1},\cA_{t+1})=(f(I_t,a_t),\cA_t\setminus a_t) \quad t=1,\ldots, T.
\end{equation}
Meanwhile the system yields a random profit $g(I_t,a_t)$.
The function $h(I_{T+1})$ is the terminal profit function at the end of the process. 

We begin with the  initial state $s_1=(I_1,\cA)$. We choose an item $a_1\in \cA$. Then the system yields a profit $g(I_1,a_1)$, and moves to the next state $s_2=(I_2,\cA_2)$ where $I_2$ follows the distribution $f(I_1,a_1)$ and $\cA_2=\cA\setminus a_1$. This process is iterated yielding a random sequence
$$
s_1,a_1,s_2,a_2,s_3,\ldots,a_{T},s_{T+1}.
$$
The profits are accumulated over $T$ steps. \footnote{If less than $T$ steps, we can use some special items to fill which satisfy that $f(I,a)=I$ and $g(I,a)=0$ for any value $I\in \cV$.}
The goal is to find a policy that maximizes the expectation of the total profits $\bE\Big[\sum_{t=1}^{T}g(I_t,a_t)+h(I_{T+1})\Big]$.
Formally, we want to determine: 
\begin{align*}
&\DP^{\ast}(s_{1})=\max_{\{a_1,\ldots,a_{T}\}\subseteq \cA}\bE\Big[\sum_{t=1}^{T}g(I_t,a_t)+h(I_{T+1})\Big]\quad\quad (\DP)\\
&\quad\quad \text{subject to: } I_{t+1}=f(I_t,a_t), \quad\quad t=1,\ldots,T.
\end{align*}
By Bellman's equation \cite{bellman1957}, for every initial state $s_1=(I_1,\cA)$, the optimal value $\DP^{\ast}(s_1)$ is given by $\DP_1(I_1,\cA)$. Here $\DP_1$ is the function defined by $\DP_{T+1}(I_{T+1})=h(I_{T+1})$ together with the recursion:
\begin{equation}\label{equ:dp}
\DP_{t}(I_t,\cA_t)=\max_{a_t\in \cA_t}\bE\Big[\DP_{t+1}(f(I_t,a_t),\cA_t\setminus a_t)+g(I_t,a_t)\Big],\quad t=1,\ldots,T.
\end{equation}
When the value and the item spaces are finite, and the expectations can be computed, this recursion yields an algorithm to compute the optimal value. However, since the state space $\cS=\cV\times2^{\cA}$ is exponentially large, this exact algorithm requires exponential time.
Since this model can capture several stochastic optimization problems which are known (or believed) be \#P-hard or even PSPACE-hard, we are interested in obtaining polynomial-time 
approximation algorithms with provable performance guarantees. 

\subsection{Our Results}
In order to obtain a polynomial time approximation scheme (PTAS) 
for the stochastic dynamic program, 
we need the following assumptions.
\begin{assump}\label{cond:main}
In this paper, we make the following assumptions.
\begin{enumerate}
\item The value space $\cV$ is discrete and ordered, and its size $|\cV|$ is a constant. W.l.o.g., we assume $\cV=(0,1,\ldots,|\cV|-1)$.
\item The function $f$ satisfies that $f(I_t,a_t)\ge I_t$, which means the value is nondecreasing. 
\item The function $h: \cV\rightarrow R^{\ge 0}$ is a nonnegative function. The expected profit $\bE[g(I_t,a_t)]$ is nonnegative (although the function $g(I_t,a_t)$ may be negative with nonzero probability).
\end{enumerate}
\end{assump}

Assumption (1) seems to be quite restrictive.
However, for several concrete problems
where the value space is not of constant size
(\eg \probemax\ in Section \ref{sec:app}), 
we can discretize the value space and reduce its size to a constant, without losing much profit. 
Assumption (2) and (3) are quite natural for many problems.
Now, we state our main result.

\begin{thm}\label{thm:main}
For any fixed $\ep>0$, if Assumption \ref{cond:main} holds, 
we can find an adaptive policy in polynomial time $n^{2^{O(\ep^{-3})}}$ with 
expected profit at least 
$\opt-O(\ep)\cdot\mopt$
where $\mopt=\mmopt$ and $\opt$ denotes the expected profit of the optimal adaptive policy. 
\end{thm}


\vspace{0.2cm}
\noindent
{\bf Our Approach:}
For the stochastic dynamic program, an optimal adaptive policy $\sigma$ can be represented as a decision tree $\cT$ (see Section \ref{sec:pre} for more details). The decision tree corresponding to the optimal policy may be exponentially large and arbitrarily complicated.
Hence, it is unlikely that one can even represent an optimal decision for the stochastic dynamic program in polynomial space. 
In order to reduce the space, we focus a special class of policies,
called {\em block adaptive policy}.
The idea of {\em block adaptive policy} was first introduced by Bhalgat\etal~\cite{bhalgat2011} and further generalized in \cite{li2013} to the context of the stochastic knapsack. To the best of our knowledge, 
the idea has not been extended to other applications.
In this paper, we make use of the notion of block adaptive policy as well, but we target at the development of a general framework.
For this sake we provide a general model of block policy (see Section \ref{sec:block}).
Since we need to work with the more abstract dynamic program,
our construction of block adaptive policy is somewhat different 
from that in \cite{bhalgat2011, li2013}.

Roughly speaking, in a block adaptive policy,
we take a batch of items simultaneously instead of a single one each time. This can significantly reduce the size of the decision tree. 
Moreover, we show that there exists a block-adaptive policy that approximates the optimal adaptive policy and has only a constant number of blocks on the decision tree (the constant depends on $\ep$). 
Since the decision tree corresponding to a block adaptive policy has a constant number of nodes, the number of all topologies of the block decision tree is a constant. Fixing the topology of the decision tree
corresponding to the block adaptive policy, we still need to decide 
the subset of items to place in each block. Again, there is exponential number of possible choices.
For each block, we can define a {\em signature} for it,
which allows us to represent a block using polynomially many 
possible signatures. 
The signatures are so defined such that two subsets with the same signature 
have approximately the same reward distribution.
Finally, we show that we can enumerate the signatures of all blocks in polynomial time using dynamic programming 
and find a nearly optimal block-adaptive policy. 
The high level idea is somewhat similar to that in \cite{li2013}, 
but the details are again quite different.



\subsection{Applications}\label{sec:app}
Our framework can be used to obtain the first PTAS
for the following problems.

\subsubsection{The \probemax~Problem}
In the \probemax~problem, 
we are given a set of $n$ items. 
Each item $i\in [n]$ has a value $X_i$ which is an independent random variable following 
a known (discrete) distribution $\pi_i$. 
We can {\em probe} a subset $P\subseteq [n]$ of items sequentially.
Each time after {\em probing} an item $i$, 
we observe its value realization, 
which is an independent sample from the distribution $\pi_i$. 
We can {\em adaptively} probe at most $m$ items and each item can be probed at most once. The reward is the maximum among the $m$ realized values. Our goal is to design an adaptive probing policy such that the expected value of the reward is maximized.

Despite being a very basic stochastic optimization problem,
we still do not have a complete understanding of the approximability of the \probemax~problem.
It is not even known whether 
it is intractable to obtain the optimal policy. 
For the non-adaptive \probemax~problem (\ie the probed set $P$ is just a priori fixed set), 
it is easy to obtain a $1-1/e$ approximation by noticing that $f(P)=\bE[\max_{i\in P}X_i]$ is a submodular 
function (see e.g., Chen \etal~\cite{chen2016}).
Chen \etal~\cite{chen2016} obtained the first PTAS.
When considering the adaptive policies, Munagala \cite{Munagala2016} provided a $\frac{1}{8}$-approximation ratio algorithm by LP relaxation. His policy is essentially a non-adaptive policy (it is 
related to the contention resolution schemes \cite{vondrak2011submodular,gupta2013}).
They also showed that the {\em adaptivity gap} (the gap between the optimal adaptive policy and optimal non-adaptive policy) is at most 3. 
For the \probemax~problem,
the best-known approximation ratio is $1-\frac{1}{e}$. 
Indeed, this can be obtained using the algorithm for 
stochastic monotone submodular maximization in
Asadpour \etal\cite{asadpour2015maximizing}. 
This is also a non-adaptive policy, 
which implies the adaptivity gap is at most $\frac{e}{e-1}$. 
In this paper, we provide the first PTAS, 
among all adaptive policies. Note that our policy is indeed adaptive. 

\begin{thm}\label{thm:probemax}
There exists a PTAS for the \probemax~problem. 
In other words, for any fixed constant $\varepsilon>0$, there is a polynomial-time approximation algorithm for the \probemax~problem that finds a policy with the expected profit at least $(1-\ep)\opt$, where $\opt$ denotes the expected profit of the optimal adaptive policy.
\end{thm}
Let the value $I_t$ be the maximum among the realized values of the probed items at the time period $t$. 
Using our framework, we have the following system dynamics for \probemax:
\begin{equation}\label{in:probemax}
I_{t+1}=f(I_t,i)=\max\{I_t,X_i\},\quad g(I_t,i)=0,\text{ and } h(I_{T+1})=I_{T+1}
\end{equation}
$t=1,2,\ldots, T$. Clearly, Assumption \ref{cond:main} (2) and (3) are satisfied. But Assumption \ref{cond:main} (1) is not satisfied because the value space $\cV$ is not of constant size.
Hence, we need to discretize the value space and reduce its size to a constant. See Section \ref{sec:prob-max} for more details.
If the reward is the summation of top-$k$ values ($k=O(1)$) among the $m$ realized values, we obtain the \probek~problem. Our techniques also allow us to derive the following result.
\begin{thm}\label{thm:topk_noncommit}
For the \probek~problem where $k$ is a constant,
there is a polynomial time algorithm that finds an adaptive policy with the expected profit at least $(1-\ep)\opt$, where $\opt$ denotes the expected profit of the optimal adaptive policy.   
\end{thm}

\subsubsection{Committed ProbeTop-k Problem}\label{subsec:committed}
We are given a set of $n$ items. Each item $i\in [n]$ has a value $X_i$ which is an independent random variable with a known (discrete) distribution $\pi_i$. 
We can {\em adaptively} probe at most $m$ items and choose $k$ values in the committed model, where $k$ is a constant.
In the {\em committed} model, once we probe an item and observe its value realization, we must make an irrevocable decision whether to choose it or not, \ie we must either add it to the final chosen set $C$ immediately or discard it forever.
\footnote{In \cite{gupta2013,gupta2016algorithms}, it is called the online decision model.} 
If we add the item to the final chosen set $C$, the realized profit is collected. Otherwise, no profit is collected and we are going the probe the next item.
Our goal is to design an adaptive probing policy such that the expected value $\bE[\sum_{i\in C}X_i]$ is maximized, where $C$ is the final chosen set.
\begin{thm}\label{thm:topk_commit}
There is a polynomial time algorithm that finds a committed policy with the expected profit at least $(1-\ep)\opt$ for the committed \probek~problem, where $\opt$ is the expected total profit
obtained by the optimal policy.
\end{thm}
Let $b_{i}^{\theta}$ represent the action that we probe
item $i$ with the threshold $\theta$ (\ie we choose item $i$ if $X_i$ realizes to a value $s$ such that $s\ge \theta$). Let $I_t$ be the the number of items that have been chosen at the period time $t$. 
Using our framework, we have following transition dynamics for the \probek~problem.
\begin{equation}
I_{t+1}=f(I_t,b_i^{\theta})=
\left\{
\begin{array}{cl}
I_t+1   & \text{if } X_i\ge \theta, I_t<k,\\
I_t & \text{otherwise; } 
\end{array}
\right.
g(I_t,b_i^{\theta})=
\left\{
\begin{array}{cl}
X_i   & \text{if } X_i\ge \theta, I_t<k,\\
0 & \text{otherwise;}
\end{array}
\right.
\end{equation}
for $t=1,2,\ldots,T$, and $h(I_{T+1})=0$.  Since $k$ is a constant, Assumption \ref{cond:main} is immediately satisfied.
There is one extra requirement for the problem:
in any realization path, we can choose at most one 
action $b_i^{\theta}$ from the set 
$\cB_i=\{b_{i}^{\theta}\}_{\theta}$.
See Section \ref{sec:probek-commit} for more details. 

\subsubsection{Committed \pandora~Problem}
For Weitzman's ``Pandora's box'' problem \cite{wei79}, we are given $n$ boxes. For each box $i\in[n]$, the probing cost $c_i$ is deterministic and the value $X_i$ is an independent random variable with a known (discrete) distribution $\pi_i$. Opening a box $i$ incurs a cost of $c_i$. When we open the box $i$, its value is realized, which is a sample from the distribution $\pi_i$. The goal is to adaptively open a subset $P\subseteq[n]$ to maximize the expected profit:
$ 
\bE\left[\max_{i\in P}\{X_i\}-\sum_{i\in P}c_i\right].
$
Weitzman provided an elegant optimal adaptive strategy,
which can be computed in polynomial time. 
Recently, Singla \cite{sahil2018} generalized this model to other combinatorial optimization problems such as matching, set cover and so on. 

In this paper, we focus on the committed model, which is mentioned in Section \ref{subsec:committed}. 
Again, we can {\em adaptively} open the boxes and choose at most $k$ values in the committed way, where $k$ is a constant.
Our goal is to design an adaptive  policy such that the expected value $\bE\left[\sum_{i\in C}X_i-\sum_{i\in P}c_i\right]$ is maximized, where $C\subseteq P$ is the final chosen set and $P$ is the set of opened boxes.
Although the problem looks like a slight variant of Weitzman's
original problem, it is quite unlikely that we can adapt
Weitzman's argument (or any argument at all) 
to obtain an optimal policy in polynomial time. 
When $k=O(1)$,
we provide the first PTAS for this problem.
Note that a PTAS is not known previously even for $k=1$.

\begin{thm}\label{thm:pandora_commit}
When $k=O(1)$, there is a polynomial time algorithm that finds a committed policy with the expected value at least $(1-\ep)\opt$ for the committed \pandora~problem. 
\end{thm}
Similar to the committed \probek~problem, let $b_{i}^{\theta}$ represent the action that we open the box $i$ with threshold $\theta$. 
Let $I_t$ be the number of boxes that have been chosen at the time period $t$. Using our framework, we have following system dynamics for the committed \pandora~problem: 
\begin{equation}
I_{t+1}=f(I_t,b_i^{\theta})=
\left\{
\begin{array}{cl}
I_t+1   & \text{if } X_i\ge \theta, I_t<k,\\
I_t & \text{otherwise; } 
\end{array}
\right.
g(I_t,b_i^{\theta})=
\left\{
\begin{array}{cl}
X_i-c_i   & \text{if } X_i\ge \theta, I_t<k,\\
-c_i & \text{otherwise;}
\end{array}
\right.
\end{equation}
for $t=1,2,\cdots,T$, and $h(I_{T+1})=0$. Notice that we never take an action $b_i^{\theta}$ for a value $I_t<k$ if $\bE[g(I_t,b_i^{\theta})]=\Pr[X_t\ge \theta]\cdot\bE[X_i\,|\,X_i\ge \theta]-c_i<0$. Then Assumption \ref{cond:main} is immediately satisfied. See Section \ref{sec:pandora} for more details. 

\eat{
\subsubsection*{4. Stochastic Knapsack \cite{dean2005adaptivity,bhalgat2011,li2013}} 
Stochastic Knapsack (SK)  was introduced by Dean \etal~\cite{dean2005adaptivity}. Suppose we have a set of items and each item is associated with a (deterministic) profit and a random size taking value with known distribution; we are given a budget and need to insert each item adaptively such that the expected profit obtained is maximized before overflow.
\eat{
We set $I_0=0, T=n$ and 
\begin{equation}
f(d_t\,|\,I_t,a_t)=\min\{I_t+d_t,2B\},\ g(d_t\,|\,I_t,a_t)=p_{t}\cdot\delta_{I_t+d_t\le B} \text{ and } h(I_{T+1})=0.
\end{equation}
}
}
\subsubsection{Stochastic Target Problem}
{\.I}lhan \etal~\cite{ilhan2011adaptive} introduced 
the following stochastic target problem. \footnote{\cite{ilhan2011adaptive} called the problem the 
	adaptive stochastic knapsack instead.
	However, their problem is quite different from 
	the stochastic knapsack problem studied in the theoretical 
	computer science literature. So we use a different name.} 
In this problem,
we are given a predetermined target $\bT$ and  a set of $n$ items.
Each item $i\in [n]$ has a value $X_i$ which is an independent random variable with a known (discrete) distribution $\pi_i$. 
Once we decide to insert an item $i$ into a knapsack, we observe a reward realization $X_i$ which follows the distribution $\pi_i$. We can insert at most $m$ items into the knapsack and our goal is to design an adaptive policy such that 
$\Pr[\sum_{i\in P}X_i\ge \bT]$ is maximized, where $P\subseteq [n]$ is the set of inserted items. 
\eat{
{\.I}lhan \etal~\cite{ilhan2011adaptive} introduced another variant of Stochastic Knapsack problem, we called Stochastic Knapsack with Random Profits (\skrp), \footnote{In~\cite{ilhan2011adaptive}, they call it Adaptive Stochastic Knapsack instead.} which has the following setting. We are given $n$ items and each item has a random reward $X_i$. Once we decide to insert an item $i$ into a knapsack, we observe a reward realization which follows a given distribution $\pi_i$ (known in advance). Consider the simple case when we can insert at most $m$ items into the knapsack. 
Our goal is to design an adaptive policy such that the probability of the total rewards of all item inserted larger or equal to a predetermined target is maximized.
}
For the stochastic target problem, {\.I}lhan \etal~\cite{ilhan2011adaptive} provided some heuristic based on dynamic programming for the special case where the random profit of each item follows a known normal distribution. 
In this paper, we provide an additive PTAS for 
the stochastic target problem when the target is relaxed to $(1-\ep)\bT$. 
\begin{thm}\label{thm:main-skrp}
There exists an additive PTAS for stochastic target problem if we relax the target to $(1-\ep)\bT$. 
In other words, for any given constant $\ep>0$, there is a polynomial-time approximation algorithm that finds a policy such that the probability of the total rewards exceeding $(1-\ep)\bT$ is at least $\opt-\ep$,  where $\opt$ is the resulting probability of an optimal adaptive policy.
\end{thm}
 Let the value $I_t$ be the total profits of the items in the knapsack at time period $t$. Using our framework, we have following system dynamics for the stochastic target problem:
\begin{equation} 
I_{t+1}=f(I_t,i)=I_t+X_i, \quad g(I_t,i)=0, \text{ and } 
h(I_{T+1})=
\left\{
\begin{array}{cc}
1 & \text{if }I_{T+1}\ge \bT,\\
0 & \text{otherwise;}
\end{array}
\right.
\end{equation}
for $t=1,2,\cdots,T$.  Then Assumption \ref{cond:main} (2,3) is immediately satisfied. But Assumption \ref{cond:main} (1) is not satisfied for that the value space $\cV$ is not of constant size. Hence, we need to discretize the value space and reduce its size to a constant. See Section \ref{sec:target} for more details. 

\subsubsection{Stochastic Blackjack Knapsack}
Levin \etal \cite{levin2014adaptivity} introduced the {\em stochastic blackjack knapsack}.
In this problem, we are given a capacity $\bC$ and a set of $n$ items, each item $i\in[n]$ has a size $X_i$ which is an independent random variable with a known distribution $\pi_i$ and a profit $p_i$.
We can adaptively insert the items into a knapsack, as long as the capacity constraint is not violated.
 Our goal is to design an adaptive policy such that the expected total profits of all items inserted is maximized. The key feature here different from classic stochastic knapsack is that we gain zero if overflow, \ie we will lose the profits of all items inserted already if the total size is larger than the capacity. This extra restriction might induce us to take more conservative policies.
Levin \etal~\cite{levin2014adaptivity} presented a non-adaptive policy with expected value that is at least $(\sqrt{2}-1)^2/2\approx 1/11.66$ times the expected value of the optimal adaptive policy.
 Chen \etal~\cite{CHEN2014625} assumed each size $X_i$ follows a known exponential distribution and gave an optimal policy for $n=2$ based on dynamic programming. In this paper, we provide the first bi-criteria PTAS for the problem.  
\begin{thm}\label{thm:main-sk-zro}
For any fixed constant $\ep>0$, there is a polynomial-time approximation algorithm for stochastic blackjack knapsack that finds a policy with the expected profit at least $(1-\ep)\opt$, when the capacity is relaxed to $(1+\ep)\bC$, where $\opt$ is the expected profit of the optimal adaptive policy. 
\end{thm}
Denote $I_t=(I_{t,1},I_{t,2})$ and let $I_{t,1},I_{t,2}$ be the total sizes and total profits of the items in the knapsack at the time period $t$ respectively. When we insert an item $i$ into the knapsack and observe its size realization, say $s_i$, we define the system dynamics function to be
\begin{equation} 
I_{t+1}=f(I_t,i)=(I_{t,1}+s_i,I_{t,2}+p_i), \quad  
h(I_{T+1})=
\left\{
\begin{array}{cc}
I_{T+1,2} & \text{if }I_{T+1,1}\le \bC,\\
0 & \text{otherwise;}
\end{array}
\right.
\end{equation}
 and $g(I_t,i)=0$ for $t=1,2,\cdots,T$.
Then Assumption \ref{cond:main} (2,3) is immediately satisfied. But Assumption \ref{cond:main} (1) is not satisfied for that the value space $\cV$ is not of constant size. Hence, we need to discretize the value space and reduce its size to a constant. See Section \ref{sec:aon} for more details.

For the case without relaxing the capacity, we can improve the result of $11.66$ in \cite{levin2014adaptivity}.
\begin{thm}\label{thm:knap_sbk_without}
For any $\ep\ge 0$, there is a polynomial time algorithm that finds a $(\frac{1}{8}-\ep)$-approximate adaptive policy for $\sbk$.
\end{thm}

\eat{
We set $I_0=0, T=n$ and 
\begin{equation}
f(d_t\,|\,I_t,a_t)=\min\{I_t+d_t,2B\},\ g(\cdot\,|\,I_t,a_t)=0 \text{ and } h(I_{T+1})=\lam\cdot\delta_{I_{T+1}\le 1}.
\end{equation}
}

\subsection{Related Work}
Stochastic dynamic program has been widely studied in
computer science and operation research (see, for example, \cite{bertsekas1995dynamic,powell2011approximate}) and has many applications in different fields. It is a natural model for decision making under uncertainty.  In 1950s, Richard Bellman \cite{bellman1957} introduced the ``principle of optimality'' 
which leads to dynamic programming algorithms for solving sequential stochastic optimization problems. 
However, Bellman's principle does not immediate lead to efficient algorithms for many problems due to  
``curse of dimensionality'' and the large state space. 

There are some constructive frameworks that provide approximation schemes for certain classes of stochastic dynamic programs.
Shmoys \etal \cite{shmoys2006approximation} dealt with stochastic linear programs. Halman \etal \cite{halman2009,halman2014fully,halman2015computationally} studies stochastic discrete DPs with scalar state and action spaces 
and designed an FPTAS for their framework. As one of the applications, they used it to solve  the stochastic ordered adaptive knapsack problem. 
As a comparison, in our model, the state space $\cS=\cV\times2^{\cA}$ is exponentially large and hence cannot be 
solved by previous framework. 


Stochastic knapsack problem $\sk$ is one of the most well-studied stochastic combinatorial optimization problem. 
We are given a knapsack of capacity $\bC$.
Each item $i\in[n]$ has a random value $X_i$ with a known distribution $\pi_i$ and a profit $p_i$. We  can adaptively insert the items to the knapsack, as long as the capacity constraint is not violated.
The goal is to maximize the expected total profit of all items inserted. 
For $\sk$, Dean \etal~\cite{dean2005adaptivity} first provide 
a constant factor approximation algorithm. Later, Bhalgat \etal~\cite{bhalgat2011} improved that ratio to $\frac{3}{8}-\ep$ and gave an algorithm with ratio of $(1-\ep)$ by using $\ep$ extra budget for any given constant $\ep\ge0$.
In that paper, the authors first introduced the notion of block adaptive policies, which is crucial for this paper. 
The best known single-criterion approximation factor is 2 
\cite{gat2011,li2013, Willma2017}. 

The \probemax~problem and \probek~problem are special cases of the general stochastic probing framework formulated by Gupta \etal \cite{gupta2016algorithms}. They showed that the adaptivity gap of any stochastic probing problem where the outer constraint is prefix-closed and the inner constraint is an intersection of $p$ matroids is at most $O(p^3\log(np))$, where $n$ is the number of items. The Bernoulli version of stochastic probing was introduced in \cite{gupta2013}, where each item $i\in U$ has a fixed value $w_i$ and is ``active'' with an independent probability $p_i$. Gupta \etal~\cite{gupta2013} presented a framework which yields a $\frac{1}{4(k^{in}+k^{out})}$-approximation algorithm for the case when  $\cI_{in}$ and  $\cI_{out}$ are respectively an intersection of $k^{in}$ and  $k^{out}$ matroids. This ratio was improved to $\frac{1}{(k^{in}+k^{out})}$ by Adamczyk \etal \cite{adamczyk2016}  using the iterative randomized rounding approach. 
Weitzman's \pandora~is a classical example in which the goal is to find out a single random variable to 
maximize the utility minus the probing cost.
Singla \cite{sahil2018} generalized this model to other combinatorial optimization problems such as matching, set cover, facility location, and obtained approximation algorithms.

\section{Policies and Decision Trees}\label{sec:pre}
An instance of stochastic dynamic program is given by $\cJ=(\cV,\cA,f,g,h,T)$. For each item $a\in \cA$ and values $I,J\in \cV$, we denote $\Phi_a(I,J):=\Pr[f(I,a)=J]$ and $\cG_a(I):=\bE[g(I,a)]$. 
The process of applying a feasible adaptive {\em policy} $\sigma$ can be represented as a decision tree $\cT_{\sigma}$. Each node $v$ on $\cT_{\sigma}$ is labeled by a unique item $a_v\in \cA$. 
Before selecting the item $a_v$, we denote the corresponding time index, the current value and the set of the remaining available items by $t_v, I_v$ and $\cA(v)$ respectively. 
Each node has several children,
each corresponding to a different value realization
(one possible $f(I_v, a_v)$).
Let $e=(v,u)$  be the $s$-th edge emanating from $s\in \cV$ where 
$s$ is the realized value.
We call $u$ the $s$-child of $v$.
Thus $e$ has probability $\pi_e:=\pi_{v,s}=\Phi_{a_v}(I_v,s)$ and weight $w_e:=s$.

We use $\bP(\sigma)$ to denote the expected profit that the policy $\sigma$ can obtain. 
For each node $v$ on $\cT_{\sigma}$, we define $\cG_v:=\cG_{a_v}(I_v)$.
In order to clearly illustrate the tree structure, we add a dummy node at the end of each root-to-leaf path and set $\cG_{v}=h(I_v)$ if $v$ is a dummy node.
Then, we recursively define the expected profit of the subtree $\cT_v$ rooted at $v$ to be
\begin{equation}\label{equ:recur}
\bP(v)=\cG_{v}+\sum_{e=(v,u)}\pi_e\cdot\bP(u),
\end{equation}
if $v$ is an internal node and $\bP(v)=\cG_{v}=h(I_v)$ if $v$ is a leaf (\ie the dummy node). The expected profit $\bP(\sigma)$ of the policy $\sigma$ is simply $\bP(\text{the root of $\cT_{\sigma}$})$. 
Then, according to Equation \eqref{equ:dp}, we have 
$$\bP(v)\le\DP_{t_v}(I_v,\cA(v))\le \DP_{1}(I_v,\cA)\le\mmopt=\mopt$$ 
for each node $v$. For a node $v$, we say the path from the root to it on $\cT_{\sigma}$ as the {\em realization path} of $v$, and denote it by $\cR(v)$. We denote the probability of reaching $v$ as $\Phi(v)=\Phi(\cR(v))=\prod_{e\in \cR(v)}\pi_e$. 
Then, we have
\begin{equation}\label{equ:policyprofit}
\bP(\sigma)=\sum_{v\in \cT_{\sigma}}\Phi(v)\cdot \cG_{v}.
\end{equation}
We use $\opt$ to denote the expected profit of the optimal adaptive policy. For each node $v$ on the tree $\cT_{\sigma}$, by Assumption \ref{cond:main}  (2) that $f(I_v,a_v)\ge I_v$, we define $\mu_v:=\Pr[f(I_v,a_v)>I_v]=1-\Phi_{a_v}(I_v,I_v)$. For a set of nodes $P$, we define $\mu(P):=\sum_{v\in P}\mu_v$.

\begin{lemma}\label{lemma:trun}
Given an policy $\sigma$, there is a  policy $\sigma'$ with profit at least $\opt-O(\ep)\cdot\mopt$ which satisfies that for any realization path $\cR$, $\mu(\cR)\le O(1/\ep)$, where $\mopt=\mmopt$.
\end{lemma}
\begin{proof}
Consider a random realization path
$\cR=(v_1,v_2,\ldots,v_{T+1})$ generated by $\sigma$.
Recall in Assumption \ref{cond:main} (1), 
the value space is $\cV=\{0,1,\cdots,|\cV|-1\}$. 
For each node $v$ on the tree, we define
$ y_v:=\bE[f(I_v,a_v)]-I_v$, which is larger than 
$$I_v\cdot\Pr[f(I_a,a_v)=I_v]+(I_v+1)\cdot\Pr[f(I_v,a_v)>I_v]-I_v=\Pr[f(I_v,a_v)>I_v]=\mu_v.
$$
We now define a sequence of random variables $\{Y_t\}_{t\in[T+1]}$:
$$
Y_t=I_t-\sum_{i=1}^{t-1}y_{v_i}.
$$
This sequence $\{Y_i\}$ is a martinale: conditioning
on current value $Y_t$, we have 
\begin{align*}
\bE[Y_{t+1}\ |\ Y_t]&=\bE\left[I_{t+1}-\sum_{i=1}^{t}y_{v_i}\ \Big|\ Y_t\right]\\
&=\bE\left[\left(I_{t}-\sum_{v=1}^{t-1}y_{v_i}\right)+I_{t+1}-I_t-y_{v_t}\ \Big|\ Y_t\right]\\
&= Y_t+\bE[I_{t+1}\ |\ Y_t]-I_t-y_{v_t}
=Y_t.
\end{align*}
The last equation is due to the definition of $y_{v_t}$.
By the martingale property, we have $\bE[Y_{T+1}]=\bE[Y_{t}]=Y_1=0$ for any $t\in[T]$.
Thus, we have 
$$|\cV|\ge\bE[I_{T+1}]=\bE\left[\sum_{i=1}^{T}y_{v_i}\right]=\bE\left[\sum_{v\in \cR}y_v\right]\ge \bE\left[\mu(\cR)\right].$$
Let $E$ be the set of realization paths $r$ on the tree for which $\mu(r)\ge 1/\ep$. Then, we have 
$\bE[\mu(\cR)]\ge \sum_{r\in E}\left[\Phi(r)\cdot\frac{1}{\ep}\right]$ which implies that $ \sum_{r\in E}\Phi(r) \le \ep\cdot \bE[\mu(\cR)]\le O(\ep)$, where $\Phi(r)$ is the probability of passing the path $r$. 
For each path $r\in E$, let $v_{r}$ be the first node on the path such that $\mu(\cR(v_{r}))\ge 1/\ep$, where $\cR(v_{r})$ is the path from the root to the node $v_r$.
Let $F$ be the set of such nodes. For the policy $\sigma$, we have a truncation on the node $v_r$ when we reach the node $v_r$, \ie we do not select items (include $v_r$) any more in the new policy $\sigma'$. The total profit loss is at most 
$$\sum_{v\in F}\left[\Phi(v)\cdot\bP(v)\right]\le\mopt\cdot\sum_{r\in E}\Phi(r)\le O(\ep)\cdot\mopt,$$
where $\mopt=\mmopt$. 
\end{proof}
W.l.o.g, we assume that all (optimal or near optimal) policies $\sigma$ considered in this paper satisfy that for any realization $\cR$, $\mu(\cR)\le O(1/\ep).$

\section{Block Adaptive Policies}\label{sec:block} 
The decision tree corresponding to the optimal policy may be exponentially large and arbitrarily complicated. Now we consider a restrict class of policies, called 
block-adaptive policy. 
The concept was first introduced by Bhalgat $\etal$ \cite{bhalgat2011} in the context of stochastic knapsack.
Our construction is somewhat different from that in \cite{bhalgat2011,li2013}.
Here, we need to define an order for each block and 
introduce the notion of approximate block policy.

Formally, a block-adaptive policy $\hat{\sigma}$ can be thought as a decision tree $\cT_{\hat{\sigma}}$. 
Each node on the tree is labeled by a {\em block} which is a set of items. 
For a block $M$, we choose an arbitrary order $\varphi$ for the items in the block. According to the order $\varphi$, we take the items one by one,  until we get a bigger value or all items in the block are taken but the value does not change
(recall from Assumption~\ref{cond:main}
that the value is nondecreasing). 
Then we visit the child block which corresponds to the realized value. 
We use $I_M$ to denote the current value right before taking the items in the block $M$. Then for each edge $e=(M,N)$, it has probability
$$
\pi^{\varphi}_{e}=\sum_{a\in M}\left[\left(\prod_{\varphi_b<\varphi_a}\Phi_b(I_M,I_M)\right)\cdot\Phi_a(I_M,I_N)\right]$$
if  $I_N>I_M$ and $\pi^{\varphi}_e=\prod_{a\in M}\Phi_a(I_M,I_M)$ if $I_N=I_M$. 

Similar to Equation \eqref{equ:recur}, for each block $M$ and an arbitrary order $\varphi$ for $M$, we recursively define the expected profit of the subtree $\cT_{M}$ rooted at $M$ to be  
\begin{equation}\label{equ:blockprofit}
\bP(M)=\cG^{\varphi}_M+\sum_{e=(M,N)}\pi_{e}^{\varphi}\cdot\bP(N)
\end{equation}
if $M$ is an internal block and $\bP(M)=h(I_M)$ if $M$ is a leaf (\ie the dummy node). 
Here $\cG^{\varphi}_M$ is the expected profit we can get from the block which is equal to
$$
\cG^{\varphi}_M
=\sum_{a\in M}\left[\left(\prod_{\varphi_b<\varphi_a}\Phi_b(I_M,I_M)\right)\cdot\cG_{a}(I_M)\right].$$
Since the profit $\cG^{\varphi}_M$ and the probability $\pi^{\varphi}_e$ are dependent on the order $\varphi$ and thus difficult to deal with,
we define the approximate block profit and the approximate probability which do not depend on the choice of the specific order $\varphi$:
\begin{equation}\label{equ:appr}
\wt{\cG}_M=\sum_{a\in M}\cG_{a}(I_M)
\quad\text{ and }\quad
\wt{\pi}_e=\sum_{a\in M}\left[ \left(\prod_{b\in M\setminus a}\Phi_b(I_M,I_M)\right)\cdot\Phi_a(I_M,I_N)\right]
\end{equation}
if $I_N>I_M$ and $\wt{\pi}_e=\prod_{a\in M}\Phi_a(I_M,I_M)$ if $I_N=I_M$. Then we recursicely define the approximate profit
\begin{equation}\label{equ:appro}
\wt{\bP}(M)=\wt{\cG}_M+\sum_{e=(M,N)}\wt{\pi}_{e}\cdot\wt{\bP}(N),
\end{equation}
if $M$ is an internal block and $\wt{\bP}(M)=\bP(M)=h(I_M)$ if $M$ is a leaf. 
For each block $M$, we define $\mu(M):=\sum_{a\in M}\left[1-\Phi_a(I_M,I_M)\right]$.
Lemma~\ref{lemma:appr} below can be used to bound the gap between the approximate profit and the original profit if the policy satisfies the following property. Then it suffices to
consider the approximate profit for a block adaptive policy $\hat{\sigma}$ in this paper.
\begin{enumerate}[(P1)]
\item Each block $M$ with more than one item satisfies that $\mu(M)\le \ep^{2}$.
\end{enumerate}

\begin{lemma}\label{lemma:appr}
For any block-adaptive policy $\hat{\sigma}$ satisfying Property (P1), we have
$$
\left(1+O(\ep^2)\right)\cdot\wt{\bP}(\hat{\sigma})\ge\bP(\hat{\sigma})\ge \left(1-\ep^{2}\right)\cdot\wt{\bP}(\hat{\sigma}).
$$
\end{lemma}

\begin{proof}
The right hand of this lemma can be proved by induction: for each block $M$ on the decision tree, we have
\begin{equation}\label{equ:app}
\bP(M)\ge (1-\ep^{2})\cdot\wt{\bP}(M).
\end{equation}
If $M$ is a leaf, we have $\bP(M)=\wt{\bP}(M)$ which implies that Equation \eqref{equ:app} holds.  For an internal block $M$, by Property (P1), we have
$$
\cG^{\varphi}_M\ge \left[\prod_{b\in M}\Phi_b(I_M,I_M)\right]\cdot \sum_{a\in M}\cG_{a}(I_M)\ge\left[1-\sum_{b\in M}\Big(1-\Phi_b(I_M,I_M)\Big)\right]\cdot \wt{\cG}_M\ge(1-\ep^{2})\cdot\wt{\cG}_M
$$
if $M$ has more than one item and $\cG^{\varphi}_M=\wt{\cG}_M$ if $M$ has only one item. For each edge $e=(M,N)$, we have $\pi^{\varphi}_e\ge \wt{\pi}_e$. Then, by induction, we have 
\begin{align*}
\bP(M)&=\cG^{\varphi}_M+\sum_{e=(M,N)}\pi^{\varphi}_{e}\cdot\bP(N)\\
&\ge (1-\ep^2)\cdot\wt{\cG}_M+\sum_{e=(M,N)}\wt{\pi}_e\cdot\left[(1-\ep^2)\cdot\wt{\bP}(N)\right]\\
&=(1-\ep^2)\cdot\wt{\bP}(M).
\end{align*}
To prove the left hand of the lemma, we use Equation \eqref{equ:policyprofit}:
$$
\bP(\hat{\sigma})=\sum_{M\in \cT_{\hat{\sigma}}}\Phi(M)\cdot \cG^{\varphi}_M
$$
where $\Phi(M)$ is the probability of reaching the block $M$.
For each edge $e=(M,N)$, if $I_M=I_N$ or $M$ has only one item, we have $\wt{\pi}_e=\pi^{\varphi}_e$.  Otherwise, we have
$$
\wt{\pi}_e\ge\left[\prod_{b\in M}\Phi_b(I_M,I_M)\right]\cdot\sum_{a\in M}\Phi_a(I_M,I_N)\ge(1-\ep^2)\cdot\sum_{a\in M}\Phi_a(I_M,I_N)\ge(1-\ep^2)\cdot\pi^{\varphi}_e.
$$
Then, for each block $M$ and its realization path $\cR(M)=(M_0,M_1,\ldots,M_{m}=M)$, we have
$$
\frac{\wt{\Phi}(M)}{\Phi(M)}=\prod_{i=0}^{m-1}\frac{\wt{\pi}_{(M_i,M_{i+1})}}{\pi^{\varphi}_{(M_i,M_{i+1})}}=\prod_{i\,:\,I_{M_i}< I_{M_{i+1}}}\frac{\wt{\pi}_{(M_i,M_{i+1})}}{\pi^{\varphi}_{(M_i,M_{i+1})}}\ge (1-\ep^2)^{|\cV|}=1-O(\ep^2),
$$
where the last inequality holds because the value is nondecreasing and $|\cV|=O(1)$. Thus we have
$$
\wt{\bP}(\hat{\sigma})=\sum_{M\in \cT_{\hat{\sigma}}}\wt{\Phi}_M\cdot \wt{\cG}_M
\ge\sum_{M\in \cT_{\hat{\sigma}}}\left[\big(1-O(\ep^2)\big)\cdot\Phi(M)\right]\cdot {\cG}_M\ge\left(1-O(\ep^2)\right)\cdot\bP(\hat{\sigma}).
$$
\end{proof}

\subsection{Constructing a Block Adaptive Policy}
In this section, we show that there exists a block-adaptive policy that approximates the optimal adaptive policy.
In order to prove this, from an optimal (or nearly optimal) adaptive policy $\sigma$, we construct a block adaptive policy $\hat{\sigma}$ which satisfies certain nice properties and 
can obtain almost as much profit as $\sigma$ does. Thus it is sufficient to restrict our search to the block-adaptive policies. 
The construction is similar to that in \cite{li2013}.
\begin{lemma}\label{lemma:block}
An optimal policy $\sigma$ can be transformed into a block adaptive policy $\hat{\sigma}$ with approximate expected profit $\wt{\bP}(\hat{\sigma})$ at least $\opt-O(\ep)\cdot\mopt$. Moreover, the block-adaptive policy $\hat{\sigma}$ satisfies Property (P1) and (P2):

\begin{enumerate}[(P1)]
	\item Each block $M$ with more than one item satisfies that $\mu(M)\le \ep^{2}$.
	\item There are at most $O(\ep^{-3})$ blocks on any root-to-leaf path on the decision tree.\end{enumerate}
\end{lemma}

\begin{figure}
\centering
\includegraphics[width=5.5in]{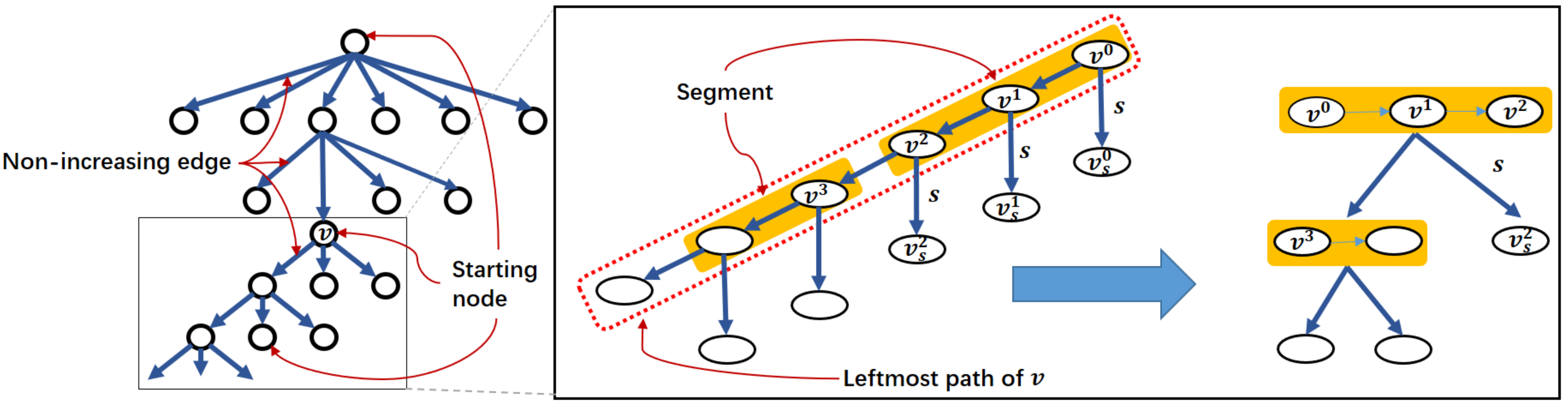}
\caption{Decision tree and block policy}\label{fig:x}
\end{figure}
\begin{proof}
For a node $v$ on the decision tree $\cT_{\sigma}$ and a value $s\in\cV$, we use $v_s$ to denote the $s$-child of $v$, which is the child of $v$ corresponding to the realized value $s$. We say an edge $e_{v,u}$ is {\em non-increasing} if $I_v=I_u$ and define the {\em leftmost path} of $v$ to be the realization path which starts at $v$, ends at a leaf, and consists of only the non-increasing edges.

We say a node $v$ is a {\em starting node} if $v$ is the root or $v$ corresponds to an increasing value of its parent $v'$ (\ie $I_v> I_{v'}$). For each staring node $v$, we  greedily partition the leftmost path of $v$ into several segments such that for any two nodes $u,w$ in the same segment $M$ and for any value $s\in \cV$, we have  
\begin{equation}\label{equ:cut}
|\bP({u_s})-\bP({w_s})|\le \ep^{2}\cdot \mopt \text{ and } \mu(M)\le \ep^{2}.
\end{equation}
Since $\mu(\cR)$ is at most $O(1/\ep)$ for each root-to-leaf path $\cR$ by Lemma \ref{lemma:trun}, the second inequality in \eqref{equ:cut} can yield at most $O(\ep^{-3})$ blocks. Now focus on the first inequality in \eqref{equ:cut}. 
Fix a particular leftmost path $\cR^{v}=(v^{0},v^{1},\ldots,v^{m})$ from a starting node $v(v=v^{0})$ on $\cT_{\sigma}$. For each value $s\in\cV$, 
we have
$$
\mopt\ge \DP_1(s,\cA)\ge\bP({v_s^{0}})\ge\bP({v_s^{1}})\ge\cdots\ge\bP({v_s^{m}})\ge 0.
$$
Otherwise, replacing the subtree $T_{v^i_s}$ with $T_{v^j_s}$ increases the profit of the policy $\sigma$ for some $i<j\le m$ if $\bP(v_s^i)<\bP(v_s^j)$. 
Thus, for each particular size $s\in \cV$, we could cut the path $\cR^{v}$ at most $\ep^{-2}$ times. Since  $|\cV|=O(1)$, we have at most $O(\ep^{-2})$ segments on the leftmost path $\cR^v$.
Now, fix a particular root-to-leaf path. Since the value is nondecreasing by Assumption \ref{cond:main} (2), there are at most $|\cV|=O(1)$ starting nodes on the path. Thus the first inequality in \eqref{equ:cut} can yield at most $O(\ep^{-2})$ segments on the root-to-leaf path.
In total, there are at most $O(\ep^{-3})$ segments on any root-to-leaf path on the decision tree.

\begin{algorithm}[h] 
 \caption{A policy $\hat{\sigma}$}
  \label{alg1}
  \begin{algorithmic}[1] \label{alg:block1}
\REQUIRE A policy ${\sigma}$.
  \STATE We start at the root of $\cT_{{\sigma}}$.
  \REPEAT
    \STATE Suppose we are at node $v$ on $\cT_{{\sigma}}$. Take the items in $\seg(v)$ one by one in the original order (the order of items in policy $\sigma$) until some node $u$ makes a transition to an increasing value, say $s$. 
   \STATE  Visit the node $l(v)_s$, the $s$-child of $l(v)$ (\ie the last node of $\seg(v)$).
	\STATE If all items in $\seg(v)$ have be taken and the value does not change, visit $l(v)_{I_v}$.
  \UNTIL {A leaf on $\cT_{{\sigma}}$ is reached.}
  \end{algorithmic}
\end{algorithm}

Now,  we are ready to describe the algorithm, which takes a policy ${\sigma}$ as input and outputs a block adaptive policy $\hat{\sigma}$. For each node $v$, we denote its segment $\seg(v)$ and use $l(v)$ to denote the last node in $\seg(v)$.
In Algorithm \ref{alg:block1}, we can see that the set of items which the policy $\hat{\sigma}$ attempts to take always corresponds to some realization path in the original policy ${\sigma}$. Property (P1) and (P2) hold immediately following from the partition argument. Now we show that the expected profit $\bP(\hat{\sigma})$ that the new policy $\hat{\sigma}$ can obtain is at least $\opt-O(\ep^2)\cdot\mopt$. 

Our algorithm deviates the policy ${\sigma}$ when the first time a node $u$ in the segment $\seg(v)$ which makes a transition to an increasing value, say $s$.  In this case, ${\sigma}$ would visit $u_s$, the $s$-child of $u$ and follows $\cT_{u_s}$ from then on. But our algorithm visits $l(v)_s$, the $s$-child of $l(v)$ (\ie the last node of $\seg(v)$), and follows $\cT_{l(v)_s}$. The expected profit gap in each such event can be bounded by
$$
\bP({u_s})-\bP({l(v)_s})\le \ep^{2}\cdot\mopt,
$$
due to the first inequality in Equation \eqref{equ:cut}. Suppose ${\sigma}$ pays such  a profit loss, and switches to visit $l(v)_s$. Then, ${\sigma}$ and our algorithm always stay at the same node. Note that
there are at most $|\cV|=O(1)$ starting nodes on any root-to-leaf path. 
 Thus ${\sigma}$ pays at most $O(1)$ times in any realization. Therefore, the total profit loss is at most $O(\ep^2)\cdot \mopt$.
By Lemma \ref{lemma:appr}, we have
$$
\wt{\bP}(\hat{\sigma})\ge\left(1-O(\ep^2)\right)\cdot\bP(\hat{\sigma})\ge\left(1-O(\ep^2)\right)\cdot \left(\opt-O(\ep^2)\cdot\mopt\right)\ge \opt-O(\ep)\cdot\mopt.
$$
\end{proof}
\subsection{Enumerating Signatures}
To search for the (nearly) optimal block-adaptive policy, we want to enumerate all possible structures of the block decision tree. Fixing the topology of the decision tree, we need to decide the subset of items to place in each block.
To do this, we define the {\em signature} such that two subsets with the same signature have approximately the same profit distribution. 
Then, we can enumerate the signatures of all blocks in polynomial time and find a nearly optimal block-adaptive policy.
Formally, 
for an item $a\in \cA$ and a value $I\in\cV=(0,1,\ldots,|\cV|-1)$, we define the {\em signature} of $a$ on $I$ to be the following vector
$$
\sig_{I}(a)=\left(\bar{\Phi}_{a}(I,0),\bar{\Phi}_{a}(I,1),\ldots,\bar{\Phi}_{a}(I,|\cV|-1),\bar{\cG}_{a}(I)\right),
$$
where
$$
\bar{\Phi}_{a}(I,J)=\left\lfloor\Phi_{a}(I,J)\cdot \frac{n}{\ep^{4}}\right\rfloor\cdot \frac{\ep^{4}}{n} 
\quad\text{and}\quad
\bar{\cG}_{a}(I)=\left\lfloor \cG_a(I)\cdot\frac{n}{\ep^{4}\mopt}\right\rfloor\cdot \frac{\ep^{4}\mopt}{n}
$$
for any $J\in \cV$. \footnote{If $\mopt=\mmopt$ is unknown, for some several concrete problems (\eg \probemax), we can get a constant approximation result for $\mopt$, which is sufficient for our purpose. In general, we can guess a constant approximation result for $\mopt$ using binary search.} For a block $M$ of items, we define the {\em signature} of $M$ on $I$ to be
$$
\sig_{I}(M)=\sum_{a\in M}\sig_{I}(a).
$$
\begin{lemma}\label{lemma:sig}
Consider two decision trees $\cT_1,\cT_2$ corresponding to block-adaptive policies with the same topology (\ie $\cT_1$ and $\cT_2$ are isomorphic) and the two block adaptive policies satisfiy Property (P1) and (P2). If for each block $M_1$ on $\cT_1$, the block $M_2$ at the corresponding position on $\cT_2$ satisfies that $\sig_{I}(M_1)=\sig_{I}(M_2)$ where $I=I_{M_1}=I_{M_2}$,
then $|\wt{\bP}(\cT_1)-\wt{\bP}(\cT_2)|\le O(\ep)\cdot\mopt$.
\end{lemma}
\begin{proof}
We focus on the case when $M$ has more than one item. Recall that for each $e=(M,N)$, we have 
$$
\wt{\pi}_e=\sum_{a\in M}\left[ \left(\prod_{b\in M\setminus a}\Phi_b(I_M,I_M)\right)\cdot\Phi_a(I_M,I_N)\right]
$$
if $I_N>I_M$ and $\wt{\pi}_e=\prod_{a\in M}\Phi_a(I_M,I_M)$ if $I_N=I_M$.
For simplicity, we use $(I,J)$ to replace $(I_M,I_N)$ if the context is clear, and write $\wt{\pi}_e$ as $\pi_M^{I}$ if $J=I$ and $\pi_M^{J}$ if $J>I$.

Fixing a block $M$, for each item $a\in M$, we define $\mu_a:=\Pr\left[f(I,a)>I\right]$.
By Property (P1) that $\mu(M)=\sum_{a\in M}[1-\Phi_a(I,I)]=\sum_{a\in M}\mu_a\le \ep^2$, we have $$
\pi^{I}_M=\prod_{a\in M}(1-\mu_a)\le\left(1-\frac{\sum_{a\in M}\mu_a}{|M|}\right)^{|M|}\!\!\!\!\!\!\!\le\exp\big(-\mu(M)\big)\!\!\le 1-\mu(M)+\mu(M)^2\le1-\mu(M)+\ep^{4}
$$
and $\pi^I_M=\prod_{a\in M}(1-\mu_a)\ge 1-\sum_{a\in M}\mu_a= 1-\mu(M)$.
Since $\sum_{a\in M}\Phi_a(I,J)\le\sum_{a\in M}\mu_a$ for any $J>I$, we have
$$
\pi^{J}_{M}=\left[\prod_{b\in M}\Phi_b(I,I)\right]\cdot\left[\sum_{a\in M}\frac{\Phi_a(I,J)}{\Phi_a(I,I)}\right]\ge (1-\ep^{2})\cdot\sum_{a\in M}\Phi_a(I,J)
\ge\sum_{a\in M}\Phi_a(I,J)-\ep^{4}.
$$
It is straightforward to verify the following property when $M$ has only one item:
$$
\pi^I_{M}=1-\mu(M) \text{ and }\pi^J_M=\sum_{a\in M}\Phi_a(I,J) \text{ for any }J>I.
$$ 
Let $M_1,M_2$ be the root blocks of $\cT_1,\cT_2$ respectively. Since $\sig_I(M_1)=\sig_I(M_2)$, we have that
$$
\left|\sum_{a\in M_1}\Phi_a(I,J)-\sum_{a\in M_2}\Phi_a(I,J)\right|\le \ep^4,
$$
for any $J\in \cV$. Then, we have
$$
\pi_{M_1}^I-\pi_{M_2}^I\le 1-\mu(M_1)+\ep^{4}-(1-\mu(M_2))=(\mu(M_2)-\mu(M_1))+\ep^{4}=O(\ep^{4}).
$$
and for any $J>I$
$$
\pi_{M_1}^J-\pi_{M_2}^J\le\sum_{a\in M_1}\Phi_a(I,J)-\sum_{a\in M_2}\Phi_a(I,J)+\ep^{4}=O(\ep^{4}).
$$ 
On the tree $\cT_1$, we replace $M_1$ with $M_2$. For each $s\in \cV$, we use $M_s$ to denote the $s$-child of block $M_1$ on $\cT_1$. Then we have 
\begin{align*}
\wt{\bP}(M_1)-\wt{\bP}(M_2)&=\left(\wt{\cG}_{M_1}-\wt{\cG}_{M_2}\right)+\sum_{s\in \cV}\wt{\bP}(M_s)\cdot\left(\pi^{s}_{M_1}-\pi^{s}_{M_2}\right)\\
&\le \ep^{4}\cdot\mopt+O(\ep^{4})\cdot\mopt\\
&=O(\ep^{4})\cdot\mopt.
\end{align*}
We replace all the blocks on $\cT_1$ by the corresponding blocks on $\cT_2$ one by one from the root to leaf. The total profit loss is at most
$
\sum_{M\in \cT_2}\left[\Phi(M)\cdot O(\ep^{4})\cdot\mopt\right]\le O(\ep)\mopt,
$
where $\Phi(M)$ is the probability of reaching $M$. The inequality holds because the depth of $\cT_2$ is at most $O(\ep^{-3})$ by Property (P2), which implies that $\sum_{M\in\cT_2}\Phi(M)\le O(\ep^{-3})$.
\end{proof}
Since $|V|=O(1)$, the number of possible signatures for a block is $O\left((n/\ep^4)^{|\cV|}\right)=n^{O(1)}$, which is a polynomial of $n$.
By Lemma \ref{lemma:block}, for any block decision tree $\cT$, there are at most $(|\cV|)^{O(\ep^{-3})}=2^{O(\ep^{-3})}$ blocks on the tree which is a constant. 
\subsection{Finding a Nearly Optimal Block-adaptive Policy} \label{sec:find-optimal}
In this section, we find a nearly optimal block-adaptive policy and prove Theorem \ref{thm:main}.
To do this, we enumerate over all topologies of the decision trees along with all possible signatures for each block.  
 This can be done by a standard dynamic programming. 

Consider a given tree topology $\cT$. A configuration $\sC$ is a set of signatures each corresponding to a block. 
Let $t_1$ and $t_2$ be the number of paths and blocks on $\cT$ respectively. We define a vector $\sCA=(u_1, u_2, \ldots, u_{t_1})$ where $u_j$ is the upper bound of the number of items on the $j$th path. For each given $i\in[n], \sC$ and $\sCA$, let 
$\dpp(i, \sC, \sCA)=1$ indicate that we can reach the configuration $\sC$ using a subset of items $\{a_1, \ldots, a_i\}$ such that the total number of items on each path $j$ is no more than $u_j$ and $0$ otherwise. Set $\dpp(0, {\bf 0}, {\bf 0})=1$ and we compute $\dpp(i, \sC,\sCA)$ in an lexicographically increasing order of $(i,\sC, \sCA)$ as follows: 
\begin{equation}\label{equ:dp-probe-max}
\dpp(i,\sC,\sCA)=\max\Big\{\dpp(i-1, \sC, \sCA ), \dpp(i-1, \sC', \sCA')\Big\}
\end{equation}
Now, we explain the above recursion as follows.
In each step, we should decide how to place the item $a_i$ on the tree $\cT$. Notice that there are at most $t_2=(|\cV|)^{O(\ep^{-3})}=2^{O(\ep^{-3})}$ blocks and therefore at most $2^{t_2}$ possible placements of item $a_i$ and each placement is called {\em feasible} if there are no two blocks on which we place the item $a_i$ have an ancestor-descendant relation.
For a feasible placement of $a_i$, we subtract $\sig(a_i)$ from each entry in $\sC$ corresponding to the block we place $a_i$ and subtract $1$ from $\sCA$ on each entry corresponding to a path including $a_i$, and in this way we get the resultant configuration $\sC'$ and  $\sCA'$ respectively. Hence, the max is over all possible such $\sC',\sCA'$.

 We have shown that the total number of all possible configurations on $\cT$ is $n^{t_2}$. The total number of vectors $\sCA$ is $T^{t_1} \le n^{t_1}\le n^{t_2}=n^{t_2}$ where $T$ is the number of rounds. For each given $(i, \sC, \sCA)$, the computation takes a constant time $O(2^{t_2})$. Thus we claim for a given tree topology, finding the optimal configuration can be done within $O(n^{2^{O(\ep^{-3})}})$ time . 
 

\begin{proof}[The proof of Theorem \ref{thm:main}]
Suppose $\sigma^{\ast}$ is the optimal policy with expected profit $\bP(\sigma^{\ast})=\opt$. We use the above dynamic programming to find a nearly optimal block adaptive policy $\sigma$.
By Lemma \ref{lemma:block}, there exists a block adaptive policy $\hat{\sigma}$ such that
$$
\wt{\bP}(\hat{\sigma})\ge \opt-O(\ep)\mopt.
$$
Since the configuration of $\hat{\sigma}$ is enumerated at some step of the algorithm, our dynamic programming is able to find a block adaptive policy $\sigma$ with the same configuration (the same tree topology and the same signatures for corresponding blocks).
By Lemma \ref{lemma:sig}, we have
$$
\wt{\bP}(\sigma)\ge \wt{\bP}(\hat{\sigma})-O(\ep)\mopt\ge \opt-O(\ep)\mopt.
$$
By Lemma \ref{lemma:appr}, we have 
$
\bP(\sigma)\ge\left(1-\ep^2\right)\cdot\wt{\bP}(\sigma)\ge \opt-O(\ep)\mopt.
$
Hence, the proof of Theorem \ref{thm:main} is completed.
\end{proof}

\section{\probemax~Problem }\label{sec:prob-max}
In this section, we demonstrate the application of our framework to the \probemax~problem. Define the value set $S=\bigcup_{i\in[n]} S_i$ where $S_i$ is the support of the random variable $X_i$ and the item set $\cA=\{1,2,\ldots,n\}$. Let the value $I_t$ be the maximum among the realized values of the probed items at the time period $t$. Thus, we begin with the initial value $I_1=0$. Since we can probe at most $m$ items, we set the number of rounds to be $T=m$. When we probe an item $i$ and observe its value realization, say $X_i$, we have the system dynamic functions
\begin{equation}\label{equ:probemax}
I_{t+1}=f(I_t,i)=\max\{I_t,X_i\},\quad g(I_t,i)=0,\text{ and } h(I_{T+1})=I_{T+1}
\end{equation}
 for $I_t\in S$ and $t=1,2,\ldots, T$. Assumption \ref{cond:main} (2,3) is immediately satisfied. But Assumption \ref{cond:main} (1) is not satisfied because the value space $S$ is not of constant size. Hence, we need discretization.  
\subsection{Discretization}
Now, we need to discretize the value space, using parameter $\ep$. We start with a
 constant factor approximate solution $\nopt$ for the \probemax~problem with 
$\opt\ge \nopt\ge (1-1/e)^2\opt$
(this can be obtained by a simple greedy algorithm
See e.g., Appendix C of \cite{chen2016}).
Let $X$ be a discrete random variable with a support $S=(s_1,s_2,\ldots,s_l)$ and $p_{s_i}=\Pr[X=s_i]$. Let $\theta=\frac{\nopt}{\ep}$ be a threshold. For ``large'' size $s_i$, \ie $s_i\ge \theta$, set $D_X(s_i)=\theta$. For ``small'' size $s_i$, \ie $s_i<\theta$, set $D_X(s_i)=\left\lfloor\frac{s_i}{\ep \nopt}\right\rfloor\ep \nopt$. 
We use $\cV=\{0,\ep \nopt,\ldots,\nopt/\ep\}$ to denote the discretized support.  Now, we describe the discretized random variable $\wt{X}$ with the support $\cV$. 
For ``large'' size, we set
\begin{equation}\label{equ:big}
\wt{p}_{\theta}=\Pr[\wt{X}=\theta]=\Pr[X\ge\theta]\cdot\frac{\bE[X\ |\ X\ge\theta]}{\theta}.
\end{equation}
Under the constraint that the sum of probabilities remains 1, for ``small'' size $d\in \cV\setminus\{\theta\}$, we scale down the probability by setting
\begin{equation}\label{equ:scale}
\wt{p}_{d}=\Pr[\wt{X}=d]=\frac{1-\Pr[\wt{X}=\theta]}{\Pr[X<\theta]}\cdot\left(\sum_{s\in S, D_X(s)=d}\Pr[X=s]\right)\le \sum_{s\in S,D_X(s)=d}\Pr[X=s].
\end{equation}

Although the above discretization is quite natural, there are some 
technical details.
We know how to solve the problem for the discretized random variables supported on $\cV$  but the realized values are
in $S$. Hence, we need to introduce the notion of 
\emph{canonical policys} 
(the notion was introduced in Bhalgat et al.~\cite{bhalgat2011}
for stochastic knapsack). 
The policy makes decisions based on the discretized sizes of variables, not their true size. More precisely, when the canonical policy $\wt{\sigma}$ probes an item $X$ which realizes to $s\in S$,  the policy makes decisions based on discretized size $D_{X}(s)$.
In this following lemma, we show it suffices to only consider canonical policies. We use $\bP(\sigma,\pi)$ to denote the expected profit that the policy $\sigma$ can obtain with the given distribution $\pi$.
\begin{lemma}\label{lemma:maxcanonical}
Let $\pi=\{\pi_i\}$ be the set of distributions of random variables and $\wt{\pi}$ be the discretized version of $\pi$. Then, we have:
\begin{enumerate}[1.]
\item For any policy $\sigma$, there exists a (canonical) policy $\wt{\sigma}$ such that
$$
\bP(\wt{\sigma},\wt{\pi})\ge(1-O(\ep))\bP(\sigma,\pi)-O(\ep)\opt;
$$

\item For any canonical policy $\wt{\sigma}$,
$$
\bP(\wt{\sigma},\pi)\ge \bP(\wt{\sigma},\wt{\pi}).
$$
\end{enumerate}
\end{lemma}

The proof of the lemma can be found in Appendix~\ref{secapp:probemax}.

\begin{proof}[The proof of Theorem \ref{thm:probemax}]
Suppose $\sigma^{\ast}$ is the optimal policy with expected profit $\bP(\sigma^{\ast},\pi)=\opt$. Given an instance $\pi$, we compute the discretized distribution $\wt{\pi}$. By Lemma \ref{lemma:maxcanonical} (1), there exists a canonical policy $\wt{\sigma}^{\ast}$ such that
$$
\bP(\wt{\sigma}^{\ast},\wt{\pi})\ge(1-O(\ep))\cdot\bP(\sigma^{\ast},\pi)-O(\ep)\opt=(1-O(\ep))\opt.
$$
Now, we present a stochastic dynamic program for the \probemax~problem with the discretized distribution $\wt{\pi}$. Define the value set $\cV=\{0,\ep \nopt,\ldots,\nopt/\ep\}$ and the item set $\cA=\{1,2,\ldots,n\}$, and set $T=m$ and $I_1=0$. When we probe an item $i$ to observe its value realization, say $X_i$, we define the system dynamic functions to be
\begin{equation}\label{equ:probemax}
I_{t+1}=f(I_t,i)=\max\{I_t,X_i\},\quad g(I_t,i)=0,\text{ and } h(I_{T+1})=I_{T+1}
\end{equation}
 for $I_t\in \cV$ and $t=1,2,\ldots, T$. Then Assumption \ref{cond:main} is immediately satisfied. By Theorem \ref{thm:main}, we can find a policy $\sigma$ with profit at least
$$
\opt_d-O(\ep^2)\cdot\mopt,
$$   
where $\opt_d$ denotes the expected profit of the optimal policy for the discretized version $\wt{\pi}$ and $\mopt=\mmopt=\DP_1(\wt{\opt}/\ep,\cA)=\wt{\opt}/\ep$.
We can see that
$
\opt_d\ge  \bP(\wt{\sigma}^{\ast},\wt{\pi})\ge (1-O(\ep))\opt.
$
Thus, by Lemma \ref{lemma:maxcanonical} (2), we have
$$
\bP(\sigma,\pi)\ge \bP(\sigma,\wt{\pi})\ge \opt_d-O(\ep^{2})\mopt\ge(1-O(\ep))\opt-O(\ep)\opt=(1-O(\ep))\opt, 
$$
which completes the proof.
\end{proof}


\subsection{ \probek~Problem}
In this section, we consider the \probek\ problem where the reward is the summation of top-$k$ values and  $k$ is a constant. 
\begin{thm}\label{thm:topkk}
There exists a PTAS for the \probek~problem. In other words, for any fixed
constant $\ep>0$, there is a polynomial-time approximation algorithm for the \probek~problem that finds a policy with the expected value at least $(1-\ep)\opt$.
\end{thm}
In this case, $I_t$ is a vector of the top-$k$ values among the realizes value of the probed items at the time period $t$. Thus, we begin with the initial vector $I_1=\{0\}^k$. When we probe an item $i$ and observe its value realization $X_i$, we update the vector by
$$
I_{t+1}=\{I_t+X_i\}\setminus \min\{I_t,X_i\}.
$$
We set $g(I_t,i)=0$ and $h(I_{T+1})= \text{sum}(I_{T+1})$. Assumption \ref{cond:main} (2,3) is immediately satisfied. Then We also need the discretization to satisfy the Assumption \ref{cond:main} (1). For Lemma \ref{lemma:maxcanonical}, we make a small change as shown in Lemma \ref{lemma:canonicaltopk}. The proof of the lemma can be found in Appendix \ref{app:probek}. 
Thus, we can prove Theorem \ref{thm:topkk} which is essentially the same as the proof of Theorem \ref{thm:probemax} and we omit it here. 
\begin{lemma}\label{lemma:canonicaltopk}
Let $\pi=\{\pi_i\}$ be the set of distributions of random variables and $\wt{\pi}$ be the discretized version of $\pi$. Then, we have:
\begin{enumerate}[1.]
\item For any policy $\sigma$, there exists a canonical policy $\wt{\sigma}$ such that
$$
\bP(\wt{\sigma},\wt{\pi})\ge(1-O(\ep))\bP(\sigma,\pi)-O(\ep)\opt;
$$

\item For any canonical policy $\wt{\sigma}$,
$$
\bP(\wt{\sigma},\pi)\ge \bP(\wt{\sigma},\wt{\pi})-O(\ep)\opt.
$$
\end{enumerate}
\end{lemma}

\eat{
\begin{lemma}\label{lemma:canonicaltopk}
Let $\pi$ be the joint distribution of random variables and $\wt{\pi}$ be the discretized version of $\pi$. Then, we have:
\begin{enumerate}[1.]
\item For any policy $\sigma$, there exists a canonical policy $\wt{\sigma}$ such that
$$
\bP(\wt{\sigma},\wt{\pi})\ge(1-O(\ep))\bP(\sigma,\pi)-O(\ep)\opt;
$$

\item For any canonical policy $\wt{\sigma}$,
$$
\bP(\wt{\sigma},\pi)\ge \bP(\wt{\sigma},\wt{\pi})-O(\ep)\opt.
$$
\end{enumerate}
\end{lemma}
}

\section{Committed \probek~Problem}\label{sec:probek-commit}
In this section, we prove Theorem \ref{thm:topk_commit}, \ie obtaining a PTAS for the committed \probek. In the committed model, once we probe an item and observe its value realization, we are committed to making an irrevocable decision immediately whether to choose it or not.  If we add the item to the final chosen set $C$, the realized profit is collected. Otherwise, no profit is collected and we are going to probe the next item.

Let $\sigma^{\ast}$ be the optimal committed policy. Suppose $\sigma^{\ast}$ is going to probe the item $i$ and choose the item $i$ if $X_i$ realizes to a value $\theta\in S_i$, where $S_i$ is the support of the random variable $X_i$. Then $\sigma^{\ast}$ would choose the item $i$ if $X_i$ realizes to a larger value $s\ge \theta$. We call $\theta$ threshold for the item $i$. 
Thus the committed policy $\sigma^{\ast}$ for the committed \probek~problem can be represented as a decision tree $\cT_{\sigma^{\ast}}$.  Every node $v$ is labeled by an unique item $a_v$ and  a threshold $\theta(v)$, which means the policy chooses the item $a_v$ if $X_v$ realizes to a size $s\ge \theta(v)$, and otherwise rejects it. 

Now, we present a stochastic dynamic program for this problem. For each item $i$, we create a set of actions $\cB_i=\{b_{i}^{\theta}\}_{\theta}$, where $b_{i}^{\theta}$ represents the action that we probe item $i$ with the threshold $\theta$. Since we  assume discrete distribution (given explicitly as the input), there are at most a polynomial number of thresholds. Hence the set of action $\cA=\cup_{i\in[n]}\cB_i$ is bounded by a polynomial. The only requirement is that at most one action from $\cB_i$ can be selected. 

Let $I_t$ be the the number of items that have been chosen at the period time $t$. Then we set $\cV=\{0,1,\ldots,k\}$, $I_1=0$. Since we can probe at most $m$ items, we set $T=m$. When we select an action $b_i^{\theta}$ to probe the item $i$ and observe its value realization, say $X_i$, we define the system dynamic functions to be
\begin{equation}
I_{t+1}=f(I_t,b_i^{\theta})=
\left\{
\begin{array}{cl}
I_t+1   & \text{if } X_i\ge \theta, I_t<k,\\
I_t & \text{otherwise; } 
\end{array}
\right.
g(I_t,b_i^{\theta})=
\left\{
\begin{array}{cl}
X_i   & \text{if } X_i\ge \theta, I_t<k,\\
0 & \text{otherwise;}
\end{array}
\right.
\end{equation}
for $I_t\in \cV$ and $t=1,2,\ldots,T$, and $h(I_{T+1})=0$.  Since $k$ is a constant, Assumption \ref{cond:main} is immediately satisfied. 
However, in this case, we cannot directly use Theorem \ref{thm:main}, due to the extra requirement that at most one action from each $\cB_i$ can be selected.
In this case, we need to slightly modify the dynamic program in Section \ref{sec:find-optimal} to satisfy the requirement.
To compute $\dpp(i,\sC,\sCA)$, once we decide the position of the item $i$, we need to choose a threshold for the item.  
Since there are at most a polynomial number of  thresholds, it can be computed at polynomial time.
Hence, again, we 
can find a policy $\sigma$ with profit at least
$$
\opt-O(\ep)\cdot\mopt=\left(1-O(\ep)\right)\opt,
$$   
where $\opt$ denotes the expected profit of the optimal policy and $\mopt=\mmopt=\DP_1(0,\cV)=\opt$. 

\section{Committed \pandora~Problem}\label{sec:pandora}
In this section, we obtain a PTAS for the committed \pandora~problem. 
This can be proved by an analogous argument to Theorem \ref{thm:topk_commit} in Section \ref{sec:probek-commit}.
Similarly, for each box $i$, we create a set of actions $\cB_i=\{b_{i}^{\theta}\}$, where $b_{i}^{\theta}$ represents the action that we open the box $i$ with threshold $\theta$. 
Let $I_t$ be the number of boxes that have been chosen at the time period $t$.
Then we set $\cA=\cup_{i\in[n]}\cB_i$, $\cV=\{0,1,\ldots,k\}$, $T=n$ and $I_1=0$. When we select an action $b_i^{\theta}$ to open the box $i$ and observe its value realization, say $X_i$, we define the system dynamic functions to be
\begin{equation}
I_{t+1}=f(I_t,b_i^{\theta})=
\left\{
\begin{array}{cl}
I_t+1   & \text{if } X_i\ge \theta, I_t<k,\\
I_t & \text{otherwise; } 
\end{array}
\right.
g(I_t,b_i^{\theta})=
\left\{
\begin{array}{cl}
X_i-c_i   & \text{if } X_i\ge \theta, I_t<k,\\
-c_i & \text{otherwise;}
\end{array}
\right.
\end{equation}
for $I_t\in \cV$ and $t=1,2,\cdots,T$, and $h(I_{T+1})=0$. Notice that we never take an action $b_i^{\theta}$ for a value $I_t<k$ if $\bE[g(I_t,b_i^{\theta})]=\Pr[X_t\ge \theta]\cdot\bE[X_i\,|\,X_i\ge \theta]-c_i<0$. Then Assumption \ref{cond:main} is immediately satisfied. 
Similar to the Committed \probek~Problem, we can choose 
at most one action from each $\cB_i$.
This can be handled in the same way.
So again we can find a policy $\sigma$ with profit at least
$
\opt-O(\ep)\cdot\mopt=\left(1-O(\ep)\right)\opt,
$   
where $\opt$ denotes the expected profit of the optimal policy and $\mopt=\mmopt=\DP_1(0,\cA)=\opt$. 


\section{Stochastic Target}\label{sec:target}
In this section, we consider the stochastic target problem and prove Theorem~\ref{thm:main-skrp}. 
Define the item set $\cA=\{1,2,\ldots,n\}$. Let the value $I_t$ be the total profits of the items in the knapsack at time period $t$. Then we set $T=m$ and $I_1=0$. When we insert an item $i$ into the knapsack and observe its profit realization, say $X_i$, we define the system dynamic functions to be
\begin{equation} 
I_{t+1}=f(I_t,i)=I_t+X_i, \quad g(I_t,i)=0, \text{ and } 
h(I_{T+1})=
\left\{
\begin{array}{cc}
1 & \text{if }I_{T+1}\ge \bT,\\
0 & \text{otherwise;}
\end{array}
\right.
\end{equation}
for $t=1,2,\cdots,T$.  Then Assumption \ref{cond:main} (2,3) is immediately satisfied. But Assumption \ref{cond:main} (1) is not satisfied for that the value space $\cV$ is not of constant size. Hence, we need discretization. 

\eat{
\begin{thm}
There exists an additive PTAS for Stochastic Knapsack with Random Profits if we relax the threshold to $1-\ep$. 
In other words, for any given constant $\ep>0$, there is a polynomial-time approximation algorithm such that the probability of the total rewards exceeding $1-\ep$ is at least $\opt-\ep$,  where $\opt$ is the resultant probability of an optimal adaptive policy.
\end{thm}
}

We use the same discretization technique as in~\cite{li2013} for the Expected Utility Maximization. The main idea is as follows. Without loss of generality, we set $\bT=1$. For an item $b$, we say $X_b$ is a big realization if $X_b>\ep^4$ and small otherwise. For a big realization of $X_b$, we simple define the discretized version of $\wt{X}_b$ as $\lfloor\frac{X_b}{\ep^5} \rfloor \ep^5$. 
For a small realization of $X_b$, we define $\wt{X}_b=0$ if $X_b<d$ and $\wt{X}_b=\ep^4$ if $d \le X_b \le \ep^4$, where $d$ is a threshold such that $\Pr[X_b \ge d\,|\, X_b \le \ep^4] \ep^4=\bE[X_b\,|\, X_b \le \ep^4]$. For more details, please refer to~\cite{li2013}. 

Let $\bP(\sigma, \pi, 1)$ be the expected objective value of the policy $\sigma$ for the instance $(\pi,1)$, where $\pi=\{\pi_i\}$ denotes the set of reward distributions and $1$ denotes the target. Let $\wt{\pi}$ be the discretized version of $\pi$.
Then, we have following lemmas. 
\begin{lemma}\label{lemma:skrp1}
For any policy $\sigma$, there exists a canonical policy $\wt{\sigma}$ such that 
$$\bP(\wt{\sigma}, \wt{\pi}, (1-2\ep))\ge\bP(\sigma, \pi, 1)-O(\ep).$$
\end{lemma}


\begin{lemma}\label{lemma:skrp2}
For any canonical policy $\wt{\sigma}$,
$$\bP(\wt{\sigma}, \pi, (1-2\ep))\ge\bP(\wt{\sigma}, \wt{\pi}, 1)-O(\ep).$$
\end{lemma}
The proof of the lemma can be found in Appendix \ref{app:skrp}.
\begin{proof}[The Proof of Theorem \ref{thm:main-skrp}]
Suppose $\sigma^{\ast}$ is the optimal policy with expected value $\opt=\bP(\sigma^{\ast},\pi,1)$. Given an instance $\pi$, we compute the discretized distribution $\wt{\pi}$. By Lemma \ref{lemma:skrp1}, there exists a policy $\wt{\sigma}^{\ast}$ such that
$$
\bP(\wt{\sigma}^{\ast}, \wt{\pi}, (1-2\ep))\ge\bP(\sigma^{\ast}, \pi, 1)-O(\ep)=\opt-O(\ep).
$$
Now, we present a stochastic dynamic  program for the instance $(\wt{\pi},1-2\ep)$. Define the value set $\cV=\{0,\ep^5,2\ep^5,\ldots,1\}$, the item set $\cA=\{1,2,\ldots,n\}$, $T=m$ and $I_1=0$. When we insert an item $i$ into the knapsack and observe its profit realization, say $X_i$, we define the system dynamic functions to be
\begin{equation} 
I_{t+1}=f(I_t,i)=\min\{1,I_t+X_i\}, \quad g(I_t,i)=0, \text{ and } 
h(I_{T+1})=
\left\{
\begin{array}{cc}
1 & \text{if }I_{T+1}\ge 1-2\ep,\\
0 & \text{otherwise;}
\end{array}
\right.
\end{equation}
for $I_t\in \cV$ and $t=1,2,\ldots,T$.  Then Assumption \ref{cond:main} is immediately satisfied. By Theorem \ref{thm:main}, we can find a policy $\sigma$ with value $\bP(\sigma, \wt{\pi},1-2\ep)$ at least
$$
\opt_d-O(\ep)\cdot\mopt\ge\bP(\wt{\sigma}^{\ast}, \wt{\pi}, (1-2\ep))
-O(\ep)=\opt-O(\ep),
$$   
where $\opt_d$ denotes the expected value of the optimal policy for instance $(\wt{\pi},1-2\ep)$ and $\mopt=\mmopt=\DP_1(1,\cA)= 1$. By Lemma \ref{lemma:skrp2}, we  have
$$
\bP(\sigma,\pi,(1-4\ep))\ge\bP(\sigma,\wt{\pi},(1-2\ep))-O(\ep)\ge \opt-O(\ep),
$$ 
which completes the proof.
\end{proof}

\section{Stochastic Blackjack Knapsack}\label{sec:aon}
In this section, we consider the stochastic blackjack knapsack and prove Theorem \ref{thm:main-sk-zro}. Define the item set $\cA=\{1,2,\ldots,n\}$. Denote $I_t=(I_{t,1},I_{t,2})$ and let $I_{t,1},I_{t,2}$ be the total sizes and total profits of the items in the knapsack at the time period $t$ respectively. We set $T=n$ and $I_1=(0,0)$. When we insert an item $i$ into the knapsack and observe its size realization, say $s_i$, we define the system dynamics function to be
\begin{equation} 
I_{t+1}=f(I_t,i)=(I_{t,1}+s_i,I_{t,2}+p_i),\ g(I_t,i)=0, \text{ and } 
h(I_{T+1})=
\left\{
\begin{array}{cc}
I_{T+1,2} & \text{if }I_{T+1,1}\le \bC,\\
0 & \text{otherwise;}
\end{array}
\right.
\end{equation}
for $t=1,2,\cdots,T$.  Then Assumption \ref{cond:main} (2,3) is immediately satisfied. But Assumption \ref{cond:main} (1) is not satisfied for that the value space $\cV$ is not of constant size. Hence, we need discretization. Unlike the stochastic traget problem, we need to discretize the sizes and profits simultaneously.

Consider a given adaptive policy $\sigma$. For each node $v\in \cT_{\sigma}$, we have $P(v)=\sum_{i\in \cR(v)}p_i$ where $\cR(v)$ is the realization path from root to $v$. 
Define $\cD=\{v\in \lf:W(v)\le \bC\}$ where $\lf$ is the set of leaves on $\cT_\sigma$. Then we have
\begin{equation}\label{equ:knap_sbk}
\bP(\sigma)=\sum_{v\in \cD}\Phi(v)\cdot P(v).
\end{equation}
Without loss of generality, we assume $\bC=1$ and $X_i \in [0,1]$ for any $i\in[n]$.
Let $\bP(\sigma, \pi, 1)$ be the expected profit of the policy $\sigma$ for the instance $(\pi,1)$, where $\pi=\{\pi_i\}$ denotes the set of size distributions and $1$ denotes the capacity. 

\subsection{Discretization}
Next, we show that item profits can be assumed to be bounded $\theta_2=\opt/\ep^2$. We set $\theta_1=\opt/\ep$ and $\theta_3=\opt/\ep^3$. Now, we define an item to be a {\em huge profit} item  if it has profit greater that or equal to $\theta_2$. We use the same discretization technique as in \cite{bhalgat2011} for the stochastic knapsack. For a huge item $b_i$ with size $X_i$ and profit $p_i$, we define a new size $\hat{X}_i$ and profit $\hat{p}_i$ as follows: for $\forall s\le 1$ 
\begin{equation}\label{equ:knap_scale}
\Pr[\hat{X}_i=s]=\Pr[\hat{X}_i=s]\cdot \frac{p_i}{\theta_2}, \quad \Pr[\hat{X}_i=1+4\ep]=1-\sum_{s\le 1}\Pr[\hat{X}_i=s]
\end{equation}
and $ \hat{p}_i=\theta_2$. In Lemma \ref{lemma:knap_profit_bound}, we show that this transformation can be performed with only an $O(\ep)$ loss in the optimal profit.
Before to prove the lemma, we need following useful lemma.
\begin{lemma}\label{lemma:knap_profit_trun} 
For any policy $\sigma$ on instance $(\pi,\bC)$, there exists a policy $\sigma'$ such that $\bP(\sigma',\pi,\bC)=(1-O(\ep))\bP(\sigma,\pi,\bC)$ and in any realization path, the sum of profit of items except the last item that $\sigma'$ inserts is less than $\theta_1$.
\end{lemma}
\begin{proof}
We interrupt the process of the policy $\sigma$ on a node $v$ when the first time that $P(v)\ge \theta_1$ to get a new policy $\sigma'$, \ie, we have a truncation on the node $v$ and do not add items (include $v$) any more in the new policy $\sigma'$. 
Let $F$ be the set of the nodes on which we have truncation. Then we have $\sum_{v\in F}\Phi(v)\le \ep$. Thus, the total profit loss is equal to
$
\sum_{v\in F}\Phi(v)\opt\le \ep\opt.
$
\end{proof}
W.l.o.g, we assume that all (optimal or near optimal) policies $\sigma$ considered in this section satisfy the following property.
\begin{enumerate}[(P3)]
\item  In any realization path, the sum of profit of items except the last item that $\sigma$ inserts is less than $\theta_1$.
\end{enumerate}
\begin{lemma}\label{lemma:knap_profit_bound} 
Let $\pi$ be the distribution of size and profit for items and $\hat{\pi}$ be the scaled version of $\pi$ by Equation \eqref{equ:knap_scale}. Then, the following statement holds:
\begin{enumerate}
\item For any policy $\sigma$, there exists a policy $\hat{\sigma}$ such that $$\bP(\hat{\sigma},\hat{\pi},\bC)=(1-O(\ep))\bP(\sigma,\pi,\bC).$$  

\item For any policy $\sigma$,
 $$\bP(\sigma,\pi,\bC)=(1-O(\ep))\bP(\sigma,\hat{\pi},\bC).$$ 
\end{enumerate}
\end{lemma}
\begin{proof}[Proof of Lemma \ref{lemma:knap_profit_bound}] 
For the first result, by Lemma \ref{lemma:knap_profit_trun}, there exists a policy $\hat{\sigma}$ such that $\bP(\hat{\sigma},\pi,\bC)=(1-O(\ep))\bP(\sigma,\pi,\bC)$ and in any realization path, there are at most one huge profit item and always at the end of the policy. For huge profit item $v$, the expected profit contributed by the realization path from root to $v$ to $\bP(\hat{\sigma},\pi,\bC)$ is 
$$
\Phi(v)\cdot \Pr[X_v\le \bC-W(v)]\cdot (P(v)+p_v).
$$
In $\bP(\hat{\sigma},\hat{\pi},\bC)$ with scaled distributions on huge profit items, the expected profit contributed by the realization path from the root to $v$ is 
\begin{align*}
&\quad\Phi(v)\cdot \Pr[\hat{X}_v\le \bC-W(v)]\cdot (P(v)+\theta_2)\\
&=\Phi(v)\cdot \left(\Pr[X_v\le \bC-W(v)]\cdot\frac{p_v}{\theta_2}\right)\cdot (P(v)+\theta_2).
\end{align*}
Since $v$ is a huge profit item, we have $p_v\ge \theta_2$, which implies $\frac{p_v}{\theta_2}\cdot (P(v)+\theta_2)\ge P(v)+p_v$. This completes the proof of the first part.

Now, we prove the second part. By Property (P3), for a huge item $v$, we have $P(v)\le \opt/\ep$. Then we have
$$
\frac{p_v}{\theta_2}\cdot (P(v)+\theta_2)
=p_v\cdot \left(1+\frac{P(v)}{\theta_2}\right)
\le p_v\cdot(1+\ep)
\le(1+\ep)(p_v+P(v)).
$$
This completes the proof of the second part.
\end{proof}

In order to discretize the profit, we define the approximate profit $\wt{\bP}(\sigma,\hat{\pi})=\sum_{v\in \cD}\Phi(v)\cdot \wt{P}(v)$ where
\begin{equation}
\wt{P}(v)=\theta_3 \cdot\left[1-\prod_{i\in \cR(v)}\left(1-\frac{p_i}{\theta_3}\right)\right]
\end{equation}
Lemma \ref{lemma:knap_skb_disprofit} below can be used to bound the gap between the approximate profit and the original profit.

\begin{lemma}\label{lemma:knap_skb_disprofit}
For any adaptive policy $\sigma$ for the scaled distribution $\hat{\pi}$, we have
$$
\bP(\sigma, \hat{\pi},\bC)\ge \wt{\bP}(\sigma,\hat{\pi},\bC)\ge (1-O(\ep))\bP(\sigma,\hat{\pi},\bC).
$$
\end{lemma}

\begin{proof}
Fix a node $v$ on the tree $\cT_{\sigma}$.
For the left side, we have
\begin{equation*}
\wt{P}(v)=\theta_3 \cdot\left[1-\prod_{i\in \cR(v)}\left(1-\frac{p_i}{\theta_3}\right)\right]
\le\theta_3 \cdot\left[1-\left(1-\sum_{i\in \cR(v)}\frac{p_i}{\theta_3}\right)\right]
=\sum_{i\in\cR(v)}p_i=P(v).
\end{equation*}
For the right size, we have
\begin{align*}
\wt{P}(v)&=\theta_3- \theta_3\cdot\left[\prod_{i\in \cR(v)}\left(1-\frac{p_i}{\theta_3}\right)\right]\\
&\ge\theta_3 -\theta_3\cdot\left[1-\sum_{i\in \cR(v)}\frac{p_i}{\theta_3}+\left(\sum_{i\in \cR(v)}\frac{p_i}{\theta_3}\right)^2\right]\\
&=\left(\sum_{i\in\cR(v)}p_i\right)\cdot \left[1-\frac{\sum_{i\in \cR(v)}p_i}{\theta_3}\right]\\
&\ge (1-O(\ep))P(v),
\end{align*}
where the last inequality holds by Property (P3) that $P(v)\le \theta_1+\theta_2$. 
\end{proof}

Now, we choose the same discretization technique to discretize the sizes with parameter $\ep^3$ which is used in Section \ref{sec:target}. For an item $b$, we say $X_b$ is a big realization if $X_b>\ep^{4\times3}$ and small otherwise. For a big realization of $X_b$, we simple define the discretized version of $\wt{X}_b$ as $\lfloor\frac{X_b}{\ep^{5\times 3}} \rfloor \ep^{5\times 3}$. 
For a small realization of $X_b$, we define $\wt{X}_b=0$ if $X_b<d$ and $\wt{X}_b=\ep^{4\times 3}$ if $d \le X_b \le \ep^{4\times 3}$, where $d$ is a threshold such that $\Pr[X_b \ge d\,|\, X_b \le \ep^{4\times 3}] \ep^{4\times 3}=\bE[X_b\,|\, X_b \le \ep^{4\times 3}]$. For more details, please refer to~\cite{li2013}.

\begin{lemma} \label{lemma:zro-canonical}
Let $\hat{\pi}$ be the distribution of size and profit for items and be $\wt{\pi}$ be the discretized version of $\hat{\pi}$. Then, the following statements holds:
\begin{enumerate}
\item For any policy $\sigma$, there exists a canonical policy $\wt{\sigma}$ such that 
$$\bP(\wt{\sigma}, \wt{\pi}, (1+2\ep))\ge(1-O(\ep))\bP(\sigma, \hat{\pi}, 1).$$
\item For any canonical policy $\wt{\sigma}$,
$$\bP(\wt{\sigma}, \hat{\pi}, (1+2\ep))\ge(1-O(\ep))\bP(\wt{\sigma}, \wt{\pi}, 1).$$
\end{enumerate}
\end{lemma}

\subsection{Proof of Theorem \ref{thm:main-sk-zro}}
Now, we ready to prove Theorem \ref{thm:main-sk-zro}.
\begin{proof}[The proof of Theorem \ref{thm:main-sk-zro}]
Suppose $\sigma^{\ast}$ is the optimal policy with expected profit $\opt=\bP(\sigma^{\ast},\pi,1)$. Given an instance $\pi$, we compute the scaled distribution $\hat{\pi}$ and discretized distribution $\wt{\pi}$. By Lemma \ref{lemma:knap_skb_disprofit}, Lemma \ref{lemma:zro-canonical} (1) and Lemma \ref{lemma:knap_profit_bound} (1),  there exist a policy $\wt{\sigma}^{\ast}$ such that
\begin{align*} 
&\quad\wt{\bP}(\wt{\sigma}^{\ast}, \wt{\pi}, (1+2\ep))\\ 
&\ge (1-O(\ep))\bP(\wt{\sigma}^{\ast}, \wt{\pi}, (1+2\ep))\quad \text{[Lemma \ref{lemma:knap_skb_disprofit}]}\\
&\ge (1-O(\ep))\bP(\sigma^{\ast},\hat{\pi}, 1)\quad \text{[Lemma \ref{lemma:zro-canonical} (1)]}\\
&\ge(1-O(\ep))\bP(\sigma^{\ast}, \pi, 1)\quad \text{[Lemma \ref{lemma:knap_profit_bound} (1)]}\\
&= (1-O(\ep))\opt.
\end{align*}
Now, we present a stochastic dynamic program for the instance $(\wt{\pi},1+2\ep)$. Define the value set $\cV=\{0,\ep^{5\times 3},2\ep^{5\times 3},\ldots,1+3\ep\}\times \{0,1\}$ and the item set $\cA=\{1,2,\ldots,n\}$. We set $T=n$ and $I_1=0$. 
When we insert an item $i$ into the knapsack, we observe its size realization $s_i$ and toss a coin to get a value $\wt{p}_i$ with $\Pr[\wt{p}_i=1]=\frac{\hat{p}_i}{\theta_3}$ and  $\Pr[\wt{p}_i=0]=1-\frac{\hat{p}_i}{\theta_3}$ . Then we define the system dynamics function to be
\begin{equation} 
I_{t+1}=f(I_t,i)=(I_{t+1,1},I_{t+1,2})=(\min\{1+3\ep,I_{t,1}+s_i\},\max\{I_{t,2},\wt{p}_i\}) 
\end{equation}
and $g(I_t,i)=0$ for $I_t\in \cV$ and $t=1,2,\cdots,T$. The terminal function is  
\begin{equation} 
h(I_{T+1})=
\left\{
\begin{array}{cc}
\theta_3\cdot I_{T+1,2} & \text{if }I_{T+1}\le 1+2\ep,\\
0 & \text{otherwise;}
\end{array}
\right.
\end{equation}
Then Assumption \ref{cond:main} is immediately satisfied. By Theorem \ref{thm:main}, we can find a policy $\sigma$ with profit  $\wt{\bP}(\sigma, \wt{\pi},1+2\ep)$ at least
$$
\opt_d-O(\ep^4)\cdot\mopt\ge\wt{\bP}(\wt{\sigma}^{\ast}, \wt{\pi}, (1+2\ep))-O(\ep)\opt=(1-O(\ep))\opt,
$$   
where $\opt_d$ denotes the expected approximate profit of the optimal policy for instance $(\wt{\pi},1+2\ep)$ and $\mopt=\mmopt=\DP_{1}((0,1),\cA)=\theta_3=\frac{\opt}{\ep^3}$. By Lemma \ref{lemma:knap_skb_disprofit}, Lemma \ref{lemma:zro-canonical} (2) and Lemma \ref{lemma:knap_profit_bound} (2), we  have
\begin{align*}
&\quad\bP(\sigma,\pi,(1+4\ep))\\
&\ge (1-O(\ep))\bP(\sigma,\hat{\pi},(1+4\ep))\quad \text{[Lemma \ref{lemma:knap_profit_bound} (2)]}\\
&\ge (1-O(\ep))\bP(\sigma,\wt{\pi},(1+2\ep))\quad \text{[Lemma \ref{lemma:zro-canonical} (2)]}\\
&\ge (1-O(\ep))\wt{\bP}(\sigma,\wt{\pi},(1+2\ep))\quad \text{[Lemma \ref{lemma:knap_skb_disprofit}]}\\
&\ge (1-O(\ep))\opt,
\end{align*}
which completes the proof.
\end{proof}

\subsection{Without Relaxing the Capacity}\label{sec:knap_sbk_without}
Before design a policy for $\sbk$ without relaxing the capacity $\bC$, we establish a connection between adaptive policies for $\sk$ and $\sbk$. 
For a particular stochastic knapsack instance $\cJ$, we use $\optskp$ to denote the expected profit of an optimal policy for stochastic knapsack. Similarly, we denote $\optsbk$ for stochastic blackjack knapsack.
Note that a policy for $\sbk$ is also a policy for $\sk$.

\begin{lemma}\label{lemma:knap_sk_sbk}
For any policy $\sigma$ for $\sk$ on instance $\cJ=(\pi,\bC)$, there exists a policy $\sigma'$ for $\sbk$ such that
\begin{equation}
\bP_{\sbk}(\sigma',\pi,\bC)\ge \frac{1}{4}\cdot\bP_{\sk}(\sigma,\pi,\bC).
\end{equation}
\end{lemma}
\begin{proof}
W.l.o.g, we assume that for any node $v\in \cT_{\sigma}$, we have $\bP(v)\le \bP_{\sk}(\sigma,\pi,\bC)$. Otherwise, we use the subtree $\cT_{v}$ to instead $\cT_{\sigma}$ for $\sk$. Set $\theta=\bP_{\sk}(\sigma,\pi,\bC)/2$. We interrupt the process of the policy $\sigma$ on a node $v$ when the first time that the summation of is larger than or equal to $\theta$ to get a new policy $\sigma'$, \ie we have a truncation on the node $v$ and do not insert the item (include $v$) any more in the new policy $\sigma'$. Let $F$ be the set of the nodes on which we have a truncation. Let $\bar{F}=\lf\setminus F$ be the set of rest leaves, where $\lf$ is the set of leaves of the tree $\cT_{\sigma'}$. Then we have
\begin{align*}
2\theta&=\sum_{v\in \bar{F}}\Phi(v)\cdot P(v)+\sum_{v\in F}\Phi(v)\cdot [P(v)+\bP(v)]\\
&\le \theta\cdot\sum_{v\in \bar{F}}\Phi(v)+\sum_{v\in F}\Phi(v) [P(v)+2\theta]\\
&= \theta+\sum_{v\in F}\Phi(v) [P(v)+\theta]\\
&\le \theta+2\sum_{v\in F}\Phi(v) P(v).
\end{align*}
Thus the expect profit of the policy $\sigma'$ for $\sbk$ is equal to $$\sum_{v\in F}\Phi(v) P(v)\ge \frac{1}{2}\cdot\theta=\frac{1}{4}\cdot
\bP_{\sk}(\sigma,\pi,\bC).$$
\end{proof}

\begin{lemma}
For any stochastic knapsack instance $\cJ$, we have
\begin{equation}
\optskp\ge\optsbk\ge \frac{1}{4}\optskp.
\end{equation}
\end{lemma}
For any fixed $\ep\ge 0$ and instance $\cJ$, by the result of \cite{gat2011}, there is a polynomial time algorithm to compute a policy $\sigma$ for $\sk$ with expected profit $(\frac{1}{2}-\ep)\optskp$. By Lemma \ref{lemma:knap_sk_sbk}, we can find a policy $\sigma'$ for $\sbk$ expected profit at least
$$
\frac{1}{4}\times (\frac{1}{2}-\ep)\optskp\ge(\frac{1}{8}-\ep)\optsbk.
$$
This completes the proof of Theorem \ref{thm:knap_sbk_without}.

\section{Concluding Remarks}
In the paper, we formally define a model based on stochastic dynamic programs. This is a generic model. There are a number of stochastic optimization problems which fit in this model. We design a polynomial time approximation schemes for this model.

We also study two important stochastic optimization problems, \probemax~problem and stochastic knapsack problem.
Using the stochastic dynamic programs, we design a PTAS for \probemax~problem, which improves the best known approximation ratio $1-1/e$. To improve the approximation ratio for \probemax~with a matroid constraint is still a open problem.
 
Next, we focus the variants of stochastic knapsack problem: stochastic blackjack knapsack and stochastic target problem. Using the stochastic dynamic programs and discretization technique, we design a PTAS for them if allowed to relax the capacity or target. To improve the ratio for stochastic knapsack problem and variants without relaxing the capacity is still a open problem.


\section*{Acknowledgements}
We would like to thank Anupam Gupta for several helpful discussions
during the various stages of the paper.
Jian Li would like to thank the Simons Institute for the Theory of
Computing, where part of this research was
carried out. Hao Fu would like to thank Sahil Singla for
useful discussions about Pandora's Box problem. 
Pan Xu would like to thank Aravind Srinivasan for his many useful comments.


\bibliographystyle{plain}
\bibliography{\jobname}

\appendix

\section{Proof of Lemma \ref{lemma:maxcanonical} }\label{secapp:probemax}

\noindent 
{\bf Lemma~\ref{lemma:maxcanonical}.}
{\em 	
Let $\pi=\{\pi_i\}$ be the set of distributions of random variables and $\wt{\pi}$ be the discretized version of $\pi$. Then, we have:
\begin{enumerate}[1.]
\item For any policy $\sigma$, there exists a (canonical) policy $\wt{\sigma}$ such that
$$
\bP(\wt{\sigma},\wt{\pi})\ge(1-O(\ep))\bP(\sigma,\pi)-O(\ep)\opt;
$$

\item For any canonical policy $\wt{\sigma}$,
$$
\bP(\wt{\sigma},\pi)\ge \bP(\wt{\sigma},\wt{\pi}).
$$
\end{enumerate}
}

\begin{proof}[Proof of Lemma \ref{lemma:maxcanonical}]
Recall that for each node $v$ on the decision tree $\cT_{\sigma}$, the value $I_v$ is the maximum among the realized value of the probed items right before probing the item $a_v$. For a path $\cR$, we use $W(\cR)$ to the denote the value of the last node on the path. 
 Let $ E$ be the set of all root-to-leaf paths in $T_{\sigma}$.
Then we have
\begin{equation}
\bP(\sigma,\pi)=\sum_{r\in  E}\Phi(r)\cdot W(r).
\end{equation}

For the first result of Lemma \ref{lemma:maxcanonical}, we prove that there is a randomized canonical policy $\sigma_{r}$ such that $\bP(\sigma_r,\wt{\pi})\ge(1-O(\ep))\bP(\sigma,\pi)-O(\ep)\opt$. Thus such a deterministic policy $\wt{\sigma}$ exists. Let  $\theta=\frac{\nopt}{\ep}$ be a threshold. We interrupt the process of the policy $\sigma$ on a node $v$ when the first time we probe an item whose weight exceeds this threshold to get a new policy $\sigma'$ 
\ie we have a truncation on the node $v$ and do not probe items (include $v$) any more in the new policy $\sigma'$. 
The total profit loss is equal to

$$
\sum_{v\in LF}\Phi(v)\cdot[\bP(v)-I_v]
\le \sum_{v\in LF}\Phi(v)\cdot \opt
= \opt\times\sum_{v\in LF}\Phi(v)
\le O(\ep)\cdot \opt,
$$
where $LF$ is the set of the nodes on which we have a truncation. The last inequality holds because $\opt\ge \sum_{v\in LF}\Phi(v)\cdot \bP(v)\ge \theta\cdot\sum_{v\in LF}\Phi(v)$.
%

The randomized policy $\sigma_r$ is derived from $\sigma'$ as follows. $\cT(\sigma_r,\wt{\pi})$ has the same tree structure as $\cT(\sigma',\pi)$. If $\sigma_r$ probes an item $\wt{X}$ and observes a discretized size $d\in \cV$, it chooses a random branch in $\cT(\sigma_r,\wt{\pi})$ among those sizes that are mapped to $d$, i.e., $\{w_e\ |\ D_X(w_e)=d\}$ according to the probability distribution
$$
\Pr[\text{branch $e$ is chosen}]=\frac{\Pr[X=w_e]}{\sum_{s\,|\,D_X(s)=d}\Pr[X=s]}.
$$
Then by Equation \eqref{equ:scale}, if $w_e<\theta$ we have
$$\wt{p}_e=\Pr[\wt{X}=d]\cdot \Pr[\text{branch $e$ is chosen}]=p_e\cdot\frac{1-\Pr[\wt{X}=\theta]}{\Pr[X<\theta]}$$
and
$$\wt{w}_e=\left\lfloor\frac{w_e}{\ep \nopt}\right\rfloor\cdot\ep \nopt\ge w_e-\ep \nopt.$$
\begin{fact}
For any node $v$ in the tree $\cT(\sigma',\pi)$ such that $I_v<\theta$, we have
\begin{equation}\label{equ:fact}
\wt{\Phi}(v)\ge (1-O(\ep))\Phi(v).
\end{equation}
\end{fact}
When we regard the path $\cR(v)$ as a policy, the expected profit of the path $\cR(v)$ can obtain is at least
$$
\theta \cdot\left[1-\prod_{i\in \cR(v)}\left(1-\Pr[\wt{X}_i=\theta]\right)\right],
$$
which is less than $\opt$, where $\cR(v)$ is the path from the root to the node $v$.
 Then we have $\prod_{i\in \cR(v)}\left(1-\Pr[\wt{X}_i=\theta]\right)\ge 1-O(\ep)$, which implies that
$$
\wt{\Phi}(v)= \Phi(v)\cdot\prod_{i\in \cR(v)}\frac{1-\Pr[\wt{X}_i=\theta]}{\Pr[X_i\le\theta]} \ge (1-O(\ep))\Phi(v).
$$
\eat{
recall $\cR(v)$ is the realization path from the root to $v$ in $\cT(\sigma,\pi)$. If $\max_{e\in \cR(v)}w_e<\theta$, we have
$$
\wt{\Phi}(v)
=\prod_{e\in \cR(v)}\wt{p}_{e}
=\ \Phi(v)\cdot \prod_{i< j}\frac{1-\Pr[\wt{X}_{v^i}=\theta]}{\Pr[X_{v^i}<\theta]},
$$
where $\cR(v)=\{v^1,\cdots,v^{j-1},v^{j}=v\}$ $(j\le m)$.
Since $$
\opt\ge\bE[\max_{i<j}X_i] \ge \bE[\max_{i< j}\wt{X}_i]\ge \theta\cdot(1-\prod_{i< j}[1-\Pr[\wt{X}_i=\theta]),$$
we have
\begin{equation}\label{equ:pro}
\wt{\Phi}(v)\ge \Phi(v)\cdot\prod_{i< j}(1-\Pr[\wt{X}_{v^i}=\theta])\ge (1-O(\ep))\Phi(v).
\end{equation}
}
Now we bound the profit that we can obtain from $\cT(\sigma_r,\wt{\pi})$. Let $E$ be the set of all root-to-leaf paths in $\cT(\sigma',\pi)$.
We split it into two parts $ E_1=\{r\in E:W(r)< \theta\}$ and  $ E_2=\{r\in E:W(r)\ge \theta\}$.
For the first part, we have
\begin{align*}
\sum_{r\in E_1}\wt{\Phi}(r)\cdot\wt{W}(r)&\ge\sum_{r\in E_1}\wt{\Phi}(r)\cdot\left[W(r)-O(\ep \nopt)\right]\\
&\ge(1-O(\ep))\left[\sum_{r\in E_1}\Phi(r)\cdot W(r)\right]-O(\ep \nopt).
\end{align*}
As mentioned before, for any path $r\in  E_2$, we interrupt the process of the policy $\sigma$ when the first time we probe an item whose weight exceeds this threshold $\theta$. We use $\ell_r$ to denote the item for path $r$. By Equation \eqref{equ:big}, we have
$
\Pr[\wt{X}=\theta]\cdot\theta=\Pr[X\ge \theta]\cdot\bE[X\ |\ X\ge \theta].
$
Then, we have
\begin{align*}
\sum_{r\in  E_2}\wt{\Phi}(r)\cdot\wt{W}(r)
&=\sum_{r\in  E_2}\wt{\Phi}(\ell_r) \cdot \Pr[\wt{X}_{\ell_r}=\theta]\cdot \theta\\
&=\sum_{r\in  E_2}\wt{\Phi}(\ell_r) \cdot  \Pr[X_{\ell_r}\ge \theta]\cdot \bE[X_{\ell_r}\ |\ X_{\ell_r}\ge \theta]\\
&\ge\sum_{r\in  E_2}(1-O(\ep))\Phi(\ell_r) \cdot  \Pr[X_{\ell_r}\ge \theta]\cdot\bE[X_{\ell_r}\ |\ X_{\ell_r}\ge \theta]\\
&=(1-O(\ep))\sum_{r\in E_2}\Phi(r)\cdot W(r).
\end{align*}
In summation, the expected profit $\bP(\sigma_r,\wt{\pi})$ is equal to
\begin{align*}
\sum_{r\in  E}\wt{\Phi}(r)\cdot\wt{W}(r)
&\ge(1-O(\ep))\sum_{r\in E}\Phi(r)\cdot W(r)-O(\ep \nopt)\\
&=(1-O(\ep))\bP(\sigma',\pi)-O(\ep)\opt\\
&=(1-O(\ep))\bP(\sigma,\pi)-O(\ep)\opt.
\end{align*}

Next, we prove the second result of Lemma \ref{lemma:maxcanonical}. Recall that a canonical policy makes decisions based on the discretized sizes. Then $\cT(\wt{\sigma},\pi)$ has the same tree structure as $\cT(\wt{\sigma},\wt{\pi})$, except that it obtain the true profit  rather than the discretized profit.
By Equation \eqref{equ:scale}, for an edge $e$ with a weight $\wt{w}_e<\theta$ on $\cT(\wt{\sigma},\wt{\pi})$, we have
$$
\pi_{e}=\sum_{s\in S:D_X(s)=\wt{w}_e}\Pr[X=s]\ge \wt{\pi}_{e}.
$$
\begin{fact}
For any node $v$ in the tree $\cT(\wt{\sigma},\wt{\pi})$ with $I_v< \theta$, we have
\begin{equation}
\wt{\Phi}(v)\le \Phi(v).
\end{equation}
\end{fact}
Similarly, we split the root-to-leaf paths set $E$ into two parts $ E_1=\{r\in E:\max_{e\in r}\wt{w}_e< \theta\}$ and  $ E_2=\{r\in E:\max_{e\in r}\wt{w}_e= \theta\}$. Then, we have
\begin{align*}
\bP(\wt{\sigma},\wt{\pi})
&=\sum_{r\in E_1}\wt{\Phi}(r)\cdot\wt{W}(r)
+\sum_{r\in  E_2}\wt{\Phi}(r)\cdot\wt{W}(r)\\
&=\sum_{r\in E_1}\wt{\Phi}(r)\cdot\wt{W}(r)+
\sum_{r\in  E_2}\wt{\Phi}(\ell_r) \cdot \Pr[\wt{X}_{\ell_r}=\theta]\cdot \theta\\
&\le\sum_{r\in E_1}\Phi(r)\cdot W(r)
+\sum_{r\in  E_2}\Phi(\ell_r) \cdot \Pr[X_{\ell_r}\ge \theta]\cdot\bE[X_{\ell_r}\ |\ X_{\ell_r}\ge \theta] \\
&=\sum_{r\in E_1}\Phi(r)\cdot W(r)
+\sum_{r\in E_2}\Phi(r)\cdot W(r)\\
&=\bP(\wt{\sigma},\pi)
\end{align*}
\end{proof}


\section{Proof of Lemma \ref{lemma:canonicaltopk}}\label{app:probek}
\noindent{Lemma \ref{lemma:canonicaltopk}.}
{\em\
Let $\pi=\{\pi_i\}$ be the set of distributions of random variables and $\wt{\pi}$ be the discretized version of $\pi$. Then, we have:
\begin{enumerate}[1.]
\item For any policy $\sigma$, there exists a canonical policy $\wt{\sigma}$ such that
$$
\bP(\wt{\sigma},\wt{\pi})\ge(1-O(\ep))\bP(\sigma,\pi)-O(\ep)\opt;
$$

\item For any canonical policy $\wt{\sigma}$,
$$
\bP(\wt{\sigma},\pi)\ge \bP(\wt{\sigma},\wt{\pi})-O(\ep)\opt.
$$
\end{enumerate}
}
\begin{proof}
This can be proved by an analogous argument as Lemma \ref{lemma:maxcanonical}. For the first result, we design a randomized canonical policy $\sigma_r$ as before. Here, $W(r)$ is the summation of the top-$k$ weights on the path $r$. For a root-to-leaf path $r$, the profit we get is equal to
\begin{equation}
W(r)=\max_{C\subseteq r,|C|\le k}\left[\sum_{i\in C}X_i\right].
\end{equation}
Now we bound the profit we can obtain from $\cT(\sigma_r,\wt{\pi})$. recall that $ E_1=\{r\in E:\max_{e\in r}w_e< \theta\}$ and  $ E_2=\{r\in E:\max_{e\in r}w_e\ge \theta\}$ where $ E$ is the set of all root-to-leaf paths. 
Then for any $r\in  E_1$, we have 
$$
\wt{W}(r)\le W(r)-k\cdot \ep \nopt=W(r)-O(\ep \nopt).
$$
For the first part, we have
$$
\sum_{r\in E_1}\wt{\Phi}(r)\cdot\wt{W}(r)
\ge(1-O(\ep))\sum_{r\in E_1}\left[\Phi(r)\cdot W(r)\right]-O(\ep \nopt).
$$
For the second part, we have
\begin{align*}
\sum_{r\in  E_2}\wt{\Phi}(r)\cdot\wt{W}(r) &=\sum_{r\in  E_2}\wt{\Phi}(\ell_r) \cdot \Pr[\wt{X}_{\ell_r}=\theta]\cdot (\theta+\wt{W}'(\ell_r))\\
&\ge\sum_{r\in  E_2}\wt{\Phi}(\ell_r) \cdot  \Pr[X_{\ell_r}\ge \theta]\cdot \left(\bE[X_{\ell_r}|X_{\ell_r}\ge \theta]+\wt{W}'(\ell_r)\right)\\
&\ge\sum_{r\in  E_2}\wt{\Phi}(\ell_r) \cdot  \Pr[X_{\ell_r}\ge \theta]\cdot\left(\bE[X_{\ell_r}|X_{\ell_r}\ge \theta]+W'(\ell_r)-O(\ep \nopt)\right)\\
&\ge(1-O(\ep))\sum_{r\in E_2}\Phi(r)\cdot W(r)-O(\ep \nopt)
\end{align*}
where $W'(r)$ is the summation of top $k-1$ weights on the path $r$. 
In summation, the expected profit $\bP(\sigma_r,\wt{\pi})$ is equal to
$$
\sum_{r\in  E}\wt{\Phi}(r)\cdot\wt{W}(r)\ge
(1-O(\ep))\bP(\sigma,\pi)-O(\ep)\opt.
$$
Now, we prove the second result. Similarly, we have
$$
\sum_{r\in E_1}\wt{\Phi}(r)\cdot\wt{W}(r)
\le\sum_{r\in E_1}\Phi(r)\cdot W(r)
$$
and
\begin{align*}
\sum_{r\in  E_2}\wt{\Phi}(r)\cdot\wt{W}(r)
&=\sum_{r\in  E_2}\wt{\Phi}(\ell_r) \cdot \Pr[\wt{X}_{\ell_r}=\theta]\cdot (\theta+\wt{W}'(\ell_r))\\
&\le\sum_{{r\in  E_2}}\Phi(\ell_r) \cdot \Pr[X_{\ell_r}\ge \theta]\cdot\bE[X_{\ell_r}\,|\,X_{\ell_r}\ge \theta]+O(\epsilon)\opt\\
&\le\sum_{r\in E_2}\Phi(r)\cdot W(r)+O(\epsilon)\opt
\end{align*}
where the first inequality holds since $\Pr[\wt{X}=\theta]\le \ep$.
Hence, the proof of the lemma is completed.
\end{proof}

\eat{
\subsection{ Non-adaptive \probek~Problem}
Now, we consider the non-adaptive version of the \probek~problem.
Define
\begin{equation}
f(P)=\bE\left[\max_{\{i_1,i_2,\ldots,i_k\}\subseteq P}\sum_{j\in[k]}{X_{i_j}}\right]
\end{equation}
Our goal is to find a subset $P\subseteq[n]$ of cardinality $m$ such that the expected reward $f(P)$ is maximized.
\begin{thm}\label{thm:nontopk}
There exists a PTAS for the non-adaptive \probek~problem.
\end{thm}
\begin{proof}
Suppose $P^{\ast}=\{X_1,X_2,\ldots,X_m\}$ is the optimal solution with expected reward $\opt$. Then there is a pseudo-adaptive policy $\sigma^{\ast}$ only using the items in the set $P^{\ast}=\{X_1,X_2,\ldots,X_m\}$. The pseudo-adaptive policy is adaptive policy which can be represented a decision tree. In fact, for any root-to-leaf path, the set of items on the path is identical. 

Lemma \ref{lemma:canonicaltopk}, Lemma \ref{lemma:block} and Lemma \ref{lemma:sig} still hold. Now, we can find a nearly optimal pseudo-block-adaptive policy only using $m$ items in total. This can be done by modifying a little bit about the dynamic program \eqref{equ:dp-probe-max}:
\begin{equation}\label{equ:dp2}
\dpp(i, \sC, \sCA)=\max \Big(\dpp(i-1, \sC, \sCA),\dpp(i-1,\sC', \sCA-1)\Big),
\end{equation}
where the capacity $\sCA$ is the number of all items used in the tree, rather than a vector capacity for each root-to-leaf path. We can use the dynamic program to find a nearly optimal pseudo-block-adaptive policy $\sigma$ (\ie a non-adaptive policy) with expect profit $(1-O(\ep))\opt$.
\end{proof}

}

\section{Proof of Lemma~\ref{lemma:skrp1} and Lemma \ref{lemma:skrp2}}
\noindent{\bf Lemma~\ref{lemma:skrp1}.}\label{app:skrp}
{\em\  
For any policy $\sigma$, there exists a canonical policy $\wt{\sigma}$ such that 
$$\bP(\wt{\sigma}, \wt{\pi}, (1-2\ep))\ge\bP(\sigma, \pi, 1)-O(\ep).$$
}

\noindent{\bf Lemma~\ref{lemma:skrp2}.}
{\em\
For any canonical policy $\wt{\sigma}$,
$$\bP(\wt{\sigma}, \pi, (1-2\ep))\ge\bP(\wt{\sigma}, \wt{\pi}, 1)-O(\ep).$$
}

Consider a given adaptive policy $\sigma$ and for each $v \in \cT_\sigma$, let $W(v)$ and $\wt{W}(v)$ be the sum of rewards on the path $\cR(v)$ before and after discretization respectively.
Recall that $\Phi(v)$ is the probability associated with the path $\cR(v)$. In the proof of Lemma 4.2 of~\cite{li2013}, it shows that for any given set $F$ of nodes in $\cT_\sigma$ which contains at most one node from each root-leaf path, our discretization has the below property:

\begin{equation} \label{eqn:ask-dis}
\sum_{v \in F: |W(v)-\wt{W}(v)| \ge 2\epsilon} \Phi(v)=O(\ep).
\end{equation}

\begin{proof}[The Proof of Lemma \ref{lemma:skrp1}]
Consider a randomized canonical policy $\wt{\sigma}$ which has the same structure as $\sigma$.  If $\sigma_r$ inserts an item $\wt{X}$ and observes a discretized size $d\in \cV$, it chooses a random branch in $\cT(\sigma_r,\wt{\pi})$ among those sizes that are mapped to $d$, i.e., $\{w_e\ |\ D_X(w_e)=d\}$ according to the probability distribution
$$
\Pr[\text{branch $e$ is chosen}]=\frac{\Pr[X=w_e]}{\sum_{s\,|\,D_X(s)=d}\Pr[X=s]}.
$$
Then, the probability of an edge on $\cT_{\sigma_r}$ is the same as that of the corresponding edge on $\cT_{\sigma}$. The only different is two edges are labels with different weight $w_e$ on $\cT_{\sigma}$ and $\wt{w}_e$ on $\cT_{\sigma_r}$.

Notice that $\bP(\wt{\sigma}, \wt{\pi}, (1-2\ep))$ is the sum of all paths $\cR(v)$ with $\wt{W}(v) \ge 1-2\epsilon$. 
Define $\cD=\{v\in \lf:W(v)\ge 1\}$ and $\wt{\cD}=\{v\in \lf:\wt{W}(v)\ge 1-2\ep\}$, where $\lf$ is the set of leaves on $\cT(\sigma,\pi)$. Therefore we have
$$\bP(\sigma, \pi, 1)=\sum_{v \in \cD} \Phi(v),~~ \bP(\wt{\sigma}, \wt{\pi}, (1-2\ep))=\sum_{v \in \wt{\cD}} \Phi(v).$$

Consider the set $\Delta_1=\cD\setminus\wt{\cD}$. For each $v \in \Delta_1$, we have $W(v) \ge 1$ and $\wt{W}(v)<1-2\ep$. Thus we claim that $|W(v)-\wt{W}(v)|>2\ep$, implying that $\Delta_1 \subseteq \Delta\doteq \{v\in\lf: |W(v)-\wt{W}(v)| \ge 2\epsilon\}$. Thus we have
$$\bP(\wt{\sigma}, \wt{\pi}, (1-2\ep)) \ge \bP(\sigma, \pi, 1)-\sum_{v \in \Delta_1} \Phi(v)
\ge \bP(\sigma, \pi, 1)-\sum_{v \in \Delta} \Phi(v) \ge  \bP(\sigma, \pi, 1)-O(\ep).
$$
\end{proof}

\begin{proof}[The Proof of Lemma \ref{lemma:skrp2}]
In our case, we focus on the decision tree $\cT(\wt{\sigma},\wt{\pi},1)$ and assume all $\wt{w}_e$ take discretized value. $\cT(\wt{\sigma},\wt{\pi},1)$ has the same tree structure as $\cT(\wt{\sigma},\pi,1-2\ep)$.

Define $\Delta_2 =\{v \in \lf, W(v) <1-2\ep,\wt{W}(v)\ge 1\}$, where $\lf$ is the set of leaves in $\cT_\sigma$. Then we see $\wt{W}(v)-W(v)>2\ep$, implying $\Delta_2 \subseteq \Delta=\{v \in \lf: |W(v)-\wt{W}(v)| \ge 2\ep\}$. By the result of Equation~\eqref{eqn:ask-dis}, we see
$$\sum_{v \in \Delta_2} \Phi(v)  \le \sum_{v \in \Delta} \Phi(v)=O(\ep).$$ 
Therefore we claim that 
$$\bP(\wt{\sigma}, \pi, (1-2\ep)) \ge \bP(\wt{\sigma}, \wt{\pi}, 1)-\sum_{v \in \Delta_2} \Phi(v) \ge \bP(\wt{\sigma}, \wt{\pi}, 1)-O(\ep). $$
\end{proof}

\section{Proof of Lemma \ref{lemma:zro-canonical}}\label{app:zro}

\noindent{
\bf Lemma \ref{lemma:zro-canonical}.}
{\em\
Let $\hat{\pi}$ be the distribution of size and profit for items and be $\wt{\pi}$ be the discretized version of $\hat{\pi}$. Then, the following statements hold:
\begin{enumerate}
\item For any policy $\sigma$, there exists a canonical policy $\wt{\sigma}$ such that 
$$\bP(\wt{\sigma}, \wt{\pi}, (1+2\ep))\ge(1-O(\ep))\bP(\sigma, \hat{\pi}, 1).$$
\item For any canonical policy $\wt{\sigma}$,
$$\bP(\wt{\sigma}, \hat{\pi}, (1+2\ep))\ge(1-O(\ep))\bP(\wt{\sigma}, \wt{\pi}, 1).$$
\end{enumerate}
}
\begin{proof}[The proof of Lemma \ref{lemma:zro-canonical}]
For the first result, consider a randomized canonical policy $\wt{\sigma}$ which has the same structure as $\sigma$.  If $\sigma_r$ inserts an item $\wt{X}$ and observes a discretized size $d\in \cV$, it chooses a random branch in $\cT(\sigma_r,\wt{\pi})$ among those sizes that are mapped to $d$, i.e., $\{w_e\ |\ D_X(w_e)=d\}$ according to the probability distribution
$$
\Pr[\text{branch $e$ is chosen}]=\frac{\Pr[X=w_e]}{\sum_{s\,|\,D_X(s)=d}\Pr[X=s]}.
$$
Then, the probability of an edge on $\cT_{\sigma_r}$ is the same as that of the corresponding edge on $\cT_{\sigma}$. The only different is two edges are labels with different weight $w_e$ on $\cT_{\sigma}$ and $\wt{w}_e$ on $\cT_{\sigma_r}$.

We have  $P(v)=\sum_{i\in \cR(v)}p_i$ which is less than $O(\opt/\ep)$ by Lemma \ref{lemma:knap_profit_bound}. 
Define $\cD=\{v\in \lf:W(v)\le 1\}$ and $\wt{\cD}=\{v\in \lf:\wt{W}(v)\le 1+2\ep\}$, where $\lf$ is the set of leaves on $\cT(\sigma,\pi)$.
Then we have
$$\bP(\sigma, \hat{\pi}, 1)=\sum_{v\in\cD} \Phi(v)\cdot P(v), \quad
\bP(\wt{\sigma}, \wt{\pi}, 1+2\ep)=\sum_{v\in \wt{\cD}} \Phi(v)\cdot P(v).$$
Define $\Delta=\{v\in \lf:|W(v)-\wt{W}(v)|\ge 2\ep\}$. Then we have $\cD\setminus \wt{\cD}\subseteq \Delta$. By the result of Equation \eqref{eqn:ask-dis}, we have 
$$\sum_{v\in \cD\setminus\wt{\cD}}\Phi(v)
\le\sum_{v \in \Delta} \Phi(v)=O(\ep^3).$$
By Property (P1), for any node $v$, we have $P(v)\le \theta_1+\theta_2$.Then the gap $\bP(\sigma, \pi, 1)-\bP(\wt{\sigma}, \wt{\pi}, (1+2\ep)$ is less than
$$
\sum_{v\in D\setminus \wt{D}}\Phi(v)\cdot P(v)\le \sum_{v\in D\setminus\wt{D}}\Phi(v)\cdot \frac{2\opt}{\ep^2}=O(\ep)\opt.
$$
This completes the proof of the first part.

Now, we prove the second part. Since a canonical policy makes decisions based on the discretized, $\cT(\wt{\sigma},\wt{\pi},1)$ has the same tree structure as $\cT(\wt{\sigma},\hat{\pi},1+2\ep)$. Define $\cD=\{v\in \lf:W(v)\le 1+2\ep\}$ and $\wt{\cD}=\{v\in \lf:\wt{W}(v)\le 1\}$, where $\lf$ is the set of leaves on $\cT(\wt{\sigma},\wt{\pi},1)$.
Then we have
$$\bP(\wt{\sigma}, \hat{\pi}, 1+2\ep)=\sum_{v\in\cD} \Phi(v) \cdot P(v), \quad
\bP(\wt{\sigma}, \wt{\pi}, 1)=\sum_{v\in \wt{\cD}} \Phi(v)\cdot P(v).$$
Then the gap $\bP(\wt{\sigma}, \wt{\pi}, 1)-\bP(\wt{\sigma}, \hat{\pi}, (1+2\ep)$ is equal to
$$
\sum_{v\in \wt{D}\setminus D}\Phi(v)\cdot P(v)\le \sum_{v\in \wt{D}\setminus D}\Phi(v)\cdot \frac{2\opt}{\ep^2}=O(\ep)\opt.
$$
\end{proof}

\end{document}